\documentclass[10pt]{article}

\usepackage[english]{babel}
\usepackage[utf8]{inputenc}
\usepackage[round]{natbib} 

\usepackage[nohead]{geometry} % page layout and dimensions
\usepackage{setspace} % set space between lines

\usepackage{hyperref} % cross-referencing
\usepackage{xpatch} % add options to hyperref format
\usepackage{url} % URL links

\usepackage{amsmath} % math formulas
\usepackage{amssymb} % math symbols
\usepackage{amsfonts} % math fonts
\usepackage{latexsym} % math symbols
\usepackage{amsthm} % theorem environment
\usepackage{mathrsfs} % math fonts
\usepackage{mathtools} % auto loads amsmath

\usepackage{booktabs} % table definition
\usepackage{multirow} % multiple rows
\usepackage{multicol} % multiple columns
\usepackage{enumitem} % numerated lists

\usepackage{graphics} % import of figures
\usepackage{caption} % figure and table captions
\usepackage{comment} % comment blocks of text
\usepackage{accents}
\usepackage{float}

\usepackage{authblk} % author 
\usepackage[table,dvipsnames]{xcolor} % color definition

\usepackage[normalem]{ulem} % \sout command

\hypersetup{
    colorlinks = {true}, 
    linkcolor = {red}, 
    urlcolor = {black}, 
    citecolor = {blue}
}

\makeatletter
\xpatchcmd\NAT@citex
{%
    \@citea\NAT@hyper@{%
    	\NAT@nmfmt{\NAT@nm}%
    	\hyper@natlinkbreak{\NAT@aysep\NAT@spacechar}{\@citeb\@extra@b@citeb}%
    	\NAT@date
    }%
}
{%
    \@citea
    \NAT@nmfmt{\NAT@nm}%
    \NAT@aysep\NAT@spacechar
    \NAT@hyper@{\NAT@date}%
}
{}{}
\xpatchcmd\NAT@citex
{%
    \@citea\NAT@hyper@{%
    	\NAT@nmfmt{\NAT@nm}%
    	\hyper@natlinkbreak{\NAT@spacechar\NAT@@open\if*#1*\else#1\NAT@spacechar\fi}%
    	{\@citeb\@extra@b@citeb}%
    	\NAT@date
    }%
}
{
    \@citea
    \NAT@nmfmt{\NAT@nm}%
    \NAT@spacechar\NAT@@open\if*#1*\else#1\NAT@spacechar\fi
    \NAT@hyper@{\NAT@date}%
}
{}{}
\makeatother

\usepackage[commentsnumbered, ruled, vlined]{algorithm2e}

\SetKwInput{algdata}{Data}
\SetKwInput{alginput}{Input}
\SetKwInput{algoutput}{Output}
\SetKwInput{alginit}{Initialize}

\usepackage{lipsum} % Dummytext
\usepackage{xargs} % Use more than one optional parameter in new commands
\usepackage[colorinlistoftodos, prependcaption, textsize = small]{todonotes}

\counterwithout{equation}{section}

\def\keywords{{\vspace{0.25cm}\noindent\textbf{Keywords:~}}}

\newtheorem{proposition}{Proposition}%[section]
\newtheorem{lemma}{Lemma}%[section]
\def\veps{\varepsilon}

\def\bbeta{\boldsymbol{\beta}}
\def\bgamma{\boldsymbol{\gamma}}

\def\bet{\boldsymbol{\eta}}

\def\bmu{\boldsymbol{\mu}}

\def\bpsi{\boldsymbol{\psi}}

\def\bGamma{\boldsymbol{\Gamma}}
\def\bDelta{\boldsymbol{\Delta}}

\def\bLambda{\boldsymbol{\Lambda}}

\def\bSigma{\boldsymbol{\Sigma}}

\def\bone{\mathbf{1}}
\def\bzero{\mathbf{0}}
\def\b0{\mathbf{0}}

\def\bu{\mathbf{u}}
\def\bv{\mathbf{v}}

\def\bx{\mathbf{x}}
\def\by{\mathbf{y}}
\def\bz{\mathbf{z}}

\def\bA{\mathbf{A}}
\def\bB{\mathbf{B}}

\def\bD{\mathbf{D}}

\def\bG{\mathbf{G}}
\def\bH{\mathbf{H}}
\def\bI{\mathbf{I}}

\def\bO{\mathbf{O}}
\def\bP{\mathbf{P}}
\def\bQ{\mathbf{Q}}
\def\bR{\mathbf{R}}
\def\bS{\mathbf{S}}

\def\bU{\mathbf{U}}
\def\bV{\mathbf{V}}
\def\bW{\mathbf{W}}
\def\bX{\mathbf{X}}
\def\bY{\mathbf{Y}}
\def\bZ{\mathbf{Z}}

\def\bsz{\boldsymbol{z}}

\def\cT{\mathcal{T}}

\def\cV{\mathcal{V}}

\def\cY{\mathcal{Y}}

\def\scb{\textsc{b}}

\def\scr{\textsc{r}}
\def\scs{\textsc{s}}

\def\scu{\textsc{u}}
\def\scv{\textsc{v}}

\def\scx{\textsc{x}}

\def\scgam{{\scriptscriptstyle\Gamma}}

\def\tB{\textrm{\textup{B}}}

\def\tU{\textrm{\textup{U}}}
\def\tV{\textrm{\textup{V}}}

\newcommand*{\dt}[1]{%
	\accentset{\mbox{\large\bfseries.}}{#1}}
\newcommand*{\ddt}[1]{%
	\accentset{\mbox{\large\bfseries.\hspace{-0.25ex}.}}{#1}}

\def\d{\textrm{\textup{d}}} % differential operator
\def\real{\mathbb{R}} % Real Numbers
\def\natural{\mathbb{N}} % Natural Numbers
\def\R{\mathbb{R}} % Real Numbers
\def\digamma{\textrm{\textup{digamma}}}	% Digamma function

\def\argmin{\mathop{\textrm{\textup{argmin}}}}

\def\vec{\textrm{\textup{vec}}}

\def\sign{\textrm{\textup{sign}}}

\def\diag{\textrm{\textup{diag}}}

\def\bzero{\mathbf{0}}
\def\bone{\mathbf{1}}

\def\E{\mathbb{E}} % Expectation
\def\Var{\mathbb{V}\textrm{\textup{ar}}} % Variance
\def\EF{\textrm{\textup{EF}}} % Powered exponential

\def\<{\langle} % Left Angular Bracket
\def\>{\rangle} % Right Angular Bracket

\makeatletter
\newsavebox{\@brx}
\newcommand{\llangle}[1][]{%
  \savebox{\@brx}{\(\m@th{#1\langle}\)}%
  \mathopen{\copy\@brx\mkern2mu\kern-0.9\wd\@brx\usebox{\@brx}}%
}
\newcommand{\rrangle}[1][]{%
  \savebox{\@brx}{\(\m@th{#1\rangle}\)}%
  \mathclose{\copy\@brx\mkern2mu\kern-0.9\wd\@brx\usebox{\@brx}}%
}
\makeatother

\def\dsty{\displaystyle}

\newcommandx{\cristiannote}[2][1=]{%
    \todo[inline, linecolor = orange,  backgroundcolor = orange!25,  bordercolor = orange,  #1]{#2}
}

\newcommandx{\davidenote}[2][1=]{%
    \todo[inline, linecolor = green, backgroundcolor = green!25, bordercolor = green, #1]{#2}
}

\newcommandx{\lievennote}[2][1=]{%
    \todo[inline, linecolor = magenta, backgroundcolor = magenta!25, bordercolor = magenta, #1]{#2}
}

\newcommandx{\alexnote}[2][1=]{%
    \todo[inline, linecolor = cyan, backgroundcolor = cyan!25, bordercolor = cyan, #1]{#2}
}

\newcommand{\cristian}[1]{\textcolor{orange}{#1}}
\newcommand{\davide}[1]{\textcolor{teal}{#1}}
\newcommand{\lieven}[1]{\textcolor{magenta}{#1}}
\newcommand{\alex}[1]{\textcolor{cyan}{#1}}

\newcommand\cristiansout{\bgroup\markoverwith{\cristian{\rule[0.5ex]{2pt}{1.0pt}}}\ULon}
\newcommand\davidesout{\bgroup\markoverwith{\davide{\rule[0.5ex]{2pt}{1.0pt}}}\ULon}
\newcommand\lievensout{\bgroup\markoverwith{\lieven{\rule[0.5ex]{2pt}{1.0pt}}}\ULon}
\newcommand\alexsout{\bgroup\markoverwith{\alex{\rule[0.5ex]{2pt}{1.0pt}}}\ULon}

\def\glm{\textrm{\textup{GLM}}}

\def\vglm{\textrm{\textup{VGLM}}}

\def\nmf{\textrm{\textup{NMF}}}

\def\glmpca{\textrm{\textup{glmPCA}}}

\def\airwls{\textrm{\textup{AIRWLS}}}

\def\sgdqn{\textrm{\textup{SGD-QN}}}
\def\adadelta{\textrm{\textup{AdaDelta}}}
\def\adagrad{\textrm{\textup{AdaGrad}}}
\def\rmsprop{\textrm{\textup{RMSProp}}}
\def\amsgrad{\textrm{\textup{AMSGrad}}}
\def\adam{\textrm{\textup{Adam}}}

\def\R{\texttt{R}}
\def\Cpp{\texttt{C++}}

\def\cmf{\texttt{CMF}}
\def\nmfpak{\texttt{NMF}}
\def\nmfp{\texttt{NMF+}}
\def\nnlm{\texttt{NNLM}}
\def\glmpca{\texttt{glmPCA}}
\def\nbwave{\texttt{NBWaVE}}

\def\airwls{\texttt{AIRWLS}}
\def\newton{\texttt{Newton}}
\def\avagrad{\texttt{AvaGrad}}
\def\fisher{\texttt{Fisher}}
\def\gfmam{\texttt{GFM-AM}}
\def\gfmvem{\texttt{GFM-VEM}}

\def\cmf{\textrm{CMF}}
\def\nmf{\textrm{NMF}}
\def\nmfp{\textrm{NMF+}}
\def\nbwave{\textrm{NBWaVE}}
\def\airwls{\textrm{AIRWLS}}
\def\newton{\textrm{Newton}}
\def\avagrad{\textrm{AvaGrad}}
\def\fisher{\textrm{Fisher}}

\def\coap{\textrm{COAP}}

\def\R{\texttt{R}}

\def\splatter{\texttt{splatter}}
\def\bioconductor{\texttt{Bioconductor}}

\def\cmfrec{\texttt{cmfrec}}

\def\NMF{\texttt{NMF}}
\def\NNLM{\texttt{NNLM}}
\def\sgdGMF{\texttt{sgdGMF}}
\def\NewWave{\texttt{NewWave}}
\def\COAP{\texttt{COAP}}
\def\GFM{\texttt{GFM}}
\def\GMF{\texttt{gmf}}
\def\DMF{\texttt{dmf}}
\def\glmPCA{\texttt{glmpca}}

\def\cmf{\textit{CMF}}
\def\nmf{\textit{NMF}}
\def\nmfp{\textit{NMF+}}
\def\nbwave{\textit{NBWaVE}}
\def\airwls{\textit{AIRWLS}}
\def\newton{\textit{Newton}}
\def\glmpca{\textit{glmPCA}}
\def\avagrad{\textit{AvaGrad}}
\def\fisher{\textit{Fisher}}
\def\nbwave{\textit{NBWaVE}}

\def\airwls{\textit{AIRWLS}}
\def\newton{\textit{Newton}}

\def\asgd{\textit{aSGD}}
\def\adam{\textit{Adam}}
\def\amsgrad{\textit{AMSGrad}}
\def\adagrad{\textit{AdaGrad}}
\def\adadelta{\textit{AdaDelta}}
\def\rmsprop{\textit{RMSProp}}
\def\sgdqn{\textit{SGD-QN}}
\def\coap{\textit{COAP}}

\def\gfmam{\textit{GFM-AM}}
\def\gfmvem{\textit{GFM-VEM}}

\def\mb{{\textsc{b}}}
\def\row{{\textsc{i}}}
\def\col{{\textsc{j}}}

\def\mb{{\scriptscriptstyle\mathcal{B}}}
\def\row{{\scriptscriptstyle\mathcal{I}}}
\def\col{{\scriptscriptstyle\mathcal{J}}}

\def\mb{{\scriptscriptstyle{B}}}
\def\row{{\scriptscriptstyle{I}}}
\def\col{{\scriptscriptstyle{J}}}

\def\nrow{n_{\row}^\star}
\def\ncol{m_{\col}^\star}

\newcommand{\cmark}{\checkmark}
\newcommand{\xmark}{$\times$}

\title{\LARGE\bf
    Stochastic gradient descent estimation of \\ 
    generalized matrix factorization models with \\
    application to single-cell RNA sequencing data
}

\author[1]{
    Cristian Castiglione%
    \thanks{%
        \href{mailto:cristian.castiglione@unibocconi.it}%
             {cristian.castiglione@unibocconi.it}, 
        \href{mailto:alexandre.segers@ugent.be}%
             {alexandre.segers@ugent.be}, 
        \href{mailto:lieven.clement@ugent.be}%
             {lieven.clement@ugent.be}, 
        \href{mailto:davide.risso@unipd.it}%
             {davide.risso@unipd.it}%
    }%
}
\author[2]{
    Alexandre Segers%
}
\author[2]{
    Lieven Clement%
}
\author[3]{
    Davide Risso%
}

\affil[1]{
    \textit{Bocconi Institute for Data Science and Analytics, Bocconi University}.
    \textit{Via R\"{o}ntgen 1, 20136 Milan, Italy}
}

\affil[2]{
    \textit{Department of Applied Mathematics, Computer Science and Statistics, Ghent University}.
    \textit{Krijgslaan 281-S9, 9000 Ghent, Belgium}
}

\affil[3]{
    \textit{Department of Statistical Sciences, University of Padova}.
    \textit{Via Cesare Battisti 241, 35121 Padova, Italy}.
}

\date{}

\begin{document}

%%% TITLE PAGE
\maketitle
\date{}

% BODY OF THE ARTICLE

%%% ABSTRACT
\begin{abstract}
\noindent
Single-cell RNA sequencing allows the quantification of gene expression at the individual cell level, enabling the study of cellular heterogeneity and gene expression dynamics.
Dimensionality reduction is a common preprocessing step critical for the visualization, clustering, and phenotypic characterization of samples. This step, often performed using principal component analysis or closely related methods, is challenging because of the size and complexity of the data.
In this work, we present a generalized matrix factorization model assuming a general exponential dispersion family distribution and we show that many of the proposed approaches in the single-cell dimensionality reduction literature can be seen as special cases of this model. Furthermore, we propose a scalable adaptive stochastic gradient descent algorithm that allows us to estimate the model efficiently, enabling the analysis of millions of cells. 
We benchmark the proposed algorithm through extensive numerical experiments against state-of-the-art methods and showcase its use in real-world biological applications.
The proposed method systematically outperforms existing methods of both generalized and non-negative matrix factorization, demonstrating faster execution times and parsimonious memory usage, while maintaining, or even enhancing, matrix reconstruction fidelity and accuracy in biological signal extraction. 
On real data, we show that our method scales seamlessly to millions of cells, enabling dimensionality reduction in large single-cell datasets.
Finally, all the methods discussed here are implemented in an efficient open-source \texttt{R} package, \texttt{sgdGMF}, available on CRAN. 

%%% KEYWORDS
\end{abstract}

%%% SECTION
\section{Introduction}%
\label{sec:introduction}

% \davide{Davide's note}
% \cristian{Cristian's note}
% \alex{Alex's note}
% \lieven{Lieven's note}

%%% SUBSECTION
\subsection{The ever-increasing single-cell RNA sequencing data sets}

Single-cell RNA sequencing (scRNA-seq) technologies have revolutionized the comprehension of biological processes by offering a quantitative measure of transcript abundance at the individual cell level. 
Single-cell resolution is critical for the study of cellular heterogeneity \citep{Wu2023,Kim2023}, temporal dynamics \citep{Kouno2013,Jean-Baptiste2019}, and cell-type differentiation \citep{Denyer2019}. However, the size and complexity of the data have dramatically increased compared to bulk sample-level assays, challenging statistical methods and software implementations to deal with thousands of genes profiled in millions of cells \citep{Angerer2017,Kharchenko2021}.

From a statistical perspective, scRNA-seq yields high-dimensional count data, in which observations (i.e., cells) lie in a $\sim 10{,}000$-dimensional gene space. Working in such high-dimensional spaces poses a wealth of statistical and computational challenges. 
\cite{Lahnemann2020} identified eleven \emph{grand challenges} in single-cell data science; here we focus on three of them: (i) the handling of data sparsity, (ii) the integration of datasets across samples and experiments, and (iii) the definition of flexible statistical frameworks for the discovery of complex gene expression patterns. These challenges are exacerbated by the growing size and complexity of single-cell datasets. Indeed, while early studies focused on few cells from one or few samples, modern studies involve complex experimental designs, collecting thousands of cells from several individuals across different experimental conditions \citep[see, e.g.,][]{Stephenson2021,Perez2022}.

Regarding sparsity, single-cell sequencing technologies yield count data with low mean and large variance, leading to extremely skewed distributions with a large fraction of zeros \citep{Hicks2018}. To address zero inflation and overdispersion problems, \cite{Risso2018} introduced a zero-inflated negative binomial framework for gene-expression matrix factorization, named ZINB-WaVE. 
Despite its superior performance compared to other methods, ZINB-WaVE's computational complexity renders it obsolete for the ever-increasing data volumes \citep{Sun2019, Cao2021}. Moreover, in modern UMI-based scRNA-seq datasets, the need for modeling zero inflation has decreased \citep{Townes2019, Svensson2020, Nguyen2023, Ahlmann2023}. 
In this regard, simpler Negative Binomial models, such as those developed by \cite{Townes2019} and \cite{Agostinis2022}, achieve similar performances with a much smaller computational burden. Nevertheless, even these methods struggle to scale with the current massive data volumes.

Another main challenge is the integration of data across different samples, experiments, or labs. These effects, often globally referred to as \emph{batches}, represent a source of unwanted variation that need to be accounted for in the analysis. 
In unsupervised settings, for instance cell clustering or trajectory analysis, batch effects are usually corrected for with ad hoc methods \citep[see, e.g.,][]{Haghverdi2018,Korsunsky2019}. While many such methods exist, the batch integration of single-cell data is still an open problem \citep{Luecken2025,Lahnemann2020}.

Factor analysis models are a promising, unified framework that allows to account for the count nature of the data, adjust for potential sources of unwanted variation, and provide a parsimonious, low-dimensional representation that can be used as input for downstream analyses. 
To this end, state-of-the-art methods leverage a low-rank representation of the data to extract useful information while correcting for known confounders such as batch effects or other technical or biological covariates \citep{Risso2018, Townes2019, Agostinis2022}.

However, because of the lack of scalability of state-of-the-art count-based factor models, biological researchers often resort to log-transformation of the scRNA-seq count tables and apply conventional principal component analysis (PCA; \citealp{Ahlmann2023}). 
While this is a much faster alternative, it ignores the complex mean-variance relation and discrete nature of the data and may introduce synthetic biases due to imperfect data transformations \citep{Townes2019}. 
Furthermore, simple PCA of log-transformed counts fails to account for batch effects. 
Therefore, the need still remains for fast and memory-efficient matrix factorization tools on the original count scale able to include covariates, such as batch effects.

\subsection{A matrix factorization perspective}
\label{subsec:matrix_factorization}

The above-mentioned factor analysis models are examples of matrix factorization, a statistical tool of fundamental importance in many theoretical and applied fields.
In general, matrix factorization methods aim at decomposing the target data matrix into the product of two lower-rank matrices explaining the principal modes of variation in the observations.
Such low-rank matrices are typically interpreted in terms of \emph{factors} and \emph{loadings}. 
Factors represent stochastic latent variables, i.e., random effects, lying in a small-dimensional space and determining the individual characteristics of each observation in the sample.
Loadings are non-stochastic coefficients mapping the latent factors into the observed data space, or some one-to-one transformation thereof.
Specifically, in single-cell RNA-seq applications, factors and loadings can be interpreted as \emph{meta-genes} and \emph{gene-weights}, respectively \citep{Brunet2004,Stein2018}.
In this context, matrix factorization is usually employed to project the cells into the space of the first few factors to then cluster them in discrete groups \citep{Kiselev2019} or infer other low-rank signal, such as pseudotime ordering \citep{Street2018}.

Principal component analysis (PCA; \citealp{jolliffe1986}), i.e., singular value decomposition (SVD), plays a central role in the literature, being the first and most used factorization method proposed in the field.
It constitutes the base for several generalizations such as probabilistic PCA \citep{Tipping1999}, factor analysis, and, more generally, generalized linear latent variable models (GLLVM, \citealp{Bartholomew2011}).
Probabilistic formulations equip PCA with a data-generating mechanism that opens the door to alternative inferential procedures, such as likelihood-based and Bayesian approaches.
Moreover, it provides a natural way to simulate new synthetic signals using a proper generating mechanism induced by the likelihood specification.
Assuming a Gaussian law for the data, PCA can be formulated as the solution to a likelihood maximization problem under appropriate identifiability constraints \citep[see, e.g.,][]{Tipping1999}. 
This perspective unveils why PCA may be suboptimal for non-Gaussian data, such as positive scores, counts, or binary observations.

Over the years, many extensions have been proposed to address the limitations of the Gaussian PCA.
Some relevant examples are non-negative matrix factorization \citep{Lee1999, Wang2012}, Binary PCA \citep{Schein2003}, Poisson PCA \citep{Durif2019, Smallman2020, Kenney2021, Virta2023}, exponential family PCA \citep{Collins2001, Mohamed2008, Li2010, Gopalan2015, Wang2020}, generalized linear latent variable models \citep{Niku2017}, generalized factor model \citep{Liu2023, Nie2024}, covariate-augmented overdispersed Poisson factor model \citep{Liu2024}, generalized PCA \citep{Townes2019}, generalized matrix factorization \citep{Kidzinski2022} and deviance matrix factorization \citep{Wang2023}.

Non-negative matrix factorization (NMF; \citealp{Lee1999, Wang2012}) decomposes positive score matrices by minimizing either the squared error loss or the Kullbak--Leibler loss under non-negativity constraints for the factor and loading matrices.
The exploitation of non-negative patterns proved successful in many applied fields, such as computer vision and recommendation systems, as well as omics feature extraction \citep{Brunet2004,Stein2018}.

Similarly, methods based on exponential family models extend PCA by assuming a more general loss function and linking the data to the latent matrix decomposition using a smooth bijective transformation. 
For instance, \cite{Collins2001} considered the exponential family Bregman divergence, while \cite{Townes2019}, \cite{Kidzinski2022} and \cite{Wang2023} considered the exponential family (negative) log-likelihood. 

Exponential family generalizations of PCA, which we refer to as \emph{generalized matrix factorization} (GMF) models, are typically estimated using alternated Fisher scoring algorithms implemented via iterative re-weighted least squares, or some modification thereof; see, e.g., \cite{Collins2001}, \cite{Kidzinski2022} and \cite{Wang2023}.
Such a procedure directly extends the classical Fisher scoring algorithm for generalized linear models \citep{McCullagh1989}, being a stable and easy-to-code algorithmic approach.
On the other hand, in high-dimensional settings, this iterative procedure suffers two major flaws: (i) it requires multiple scans of the entire dataset during each iteration; (ii) it requires the costly numerical solution of several linear systems. 
It is worth noting that, for matrix factorization problems, the number of parameters to update, i.e. the number of linear systems to solve, is proportional to the dimension of the data matrix.
To address points (i) and (ii) in the context of scRNA-seq, \cite{Townes2019} employed an alternated Fisher scoring method which only requires the computation of element-wise derivatives and matrix multiplications, avoiding expensive matrix inversions.
In the same vein, \cite{Kidzinski2022} proposed a quasi-Newton algorithm that only requires cheap element-wise algebraic calculations.
To the best of our knowledge, these are the most efficient methods proposed in the literature for the estimation of matrix factorization models under exponential family likelihood.  
However, in both cases, a complete pass through the data is still necessary to completely update the parameter estimates, which might be infeasible in massive data scenarios.

\subsection{Our contribution}

In this work, we propose a scalable stochastic optimization algorithm to tackle the complex optimization problem underlying the estimation of high-dimensional GMF models.
Specifically, the proposed algorithm relies on the stochastic gradient descent (SGD) framework \citep{Robbins1951, Bottou2010} with adaptive learning rate schedules \citep{Duchi2011, Zeiler2012, Kingma2014, Reddi2019}.
By doing this, we decrease the computational complexity of the problem by using a convenient combination of minibatch subsampling, partial parameter updates and exponential gradient averaging.
Additionally, we propose two efficient initialization methods, which help the convergence to a meaningful solution while reducing the likelihood of the algorithm's convergence to highly sub-optimal stationary points.

\begin{table}
    \footnotesize
    \caption{
        \label{tab:package_survey}
        List of the major matrix factorization models available in the literature for exponential family data.
        For each model, we report the corresponding \R~package (first column) and we describe its characteristics.
        Each feature is marked with \cmark~if it is completely implemented in the package and \xmark~otherwise.
        The column \emph{Model} indicates the broad model family implemented in the package, where GMF and NMF stand for generalized matrix factorization, and non-negative matrix factorization, respectively.
        The column \emph{Families} lists the most common distributions belonging to the exponential family along with some generalizations: Normal (N), Gamma (G), Binomial (B), Poisson (P), Negative Binomial (NB), quasi-likelihood (Q) and zero-inflated (ZI) models.
        The column \emph{Effects} refers to the regression effects that can be included in the linear predictor and uses the notation introduced in Section \ref{sec:model_specification}, Equation \eqref{eq:linear_predictor_model}.
        The column \emph{Implementation} describes some technical features of the numerical implementation, such as the language used for the core computations (\emph{Core}), if automatic missing value estimation is allowed (\emph{Missing}), if parallel computing is allowed (\emph{Parallel}), and if minibatch subsample and stochastic optimization methods are available (\emph{Stochastic}).
        }
    \begin{tabular}{ll ccccccc ccc cccc}    
        \toprule
        {\bf Package} & {\bf Model} & \multicolumn{7}{c}{\bf Families} & \multicolumn{3}{c}{\bf Effects} & \multicolumn{4}{c}{\bf Implementation} \\
        \cmidrule(lr){3-9} \cmidrule(lr){10-12} \cmidrule(lr){13-16}
        & & N & G & B & P & NB & Q & ZI & $\bX \bB^{\scriptscriptstyle \top}$ & $\bGamma \bZ^{\scriptscriptstyle \top}$ & $\bU \bV^{\scriptscriptstyle \top}$ & Core & Missing & Parallel & Stochastic \\
        \midrule
        \texttt{RSpectra} & PCA & \cmark & \xmark & \xmark & \xmark & \xmark & \xmark & \xmark & \xmark & \xmark & \cmark & \texttt{F}/\Cpp & \xmark & \xmark & \cmark \\
        \texttt{NMF}      & NMF & \cmark & \xmark & \xmark & \cmark & \xmark & \xmark & \xmark & \xmark & \xmark & \cmark & \Cpp & \xmark & \cmark & \xmark \\
        \texttt{NNLM}     & NMF & \cmark & \xmark & \xmark & \cmark & \xmark & \xmark & \xmark & \xmark & \xmark & \cmark & \Cpp & \cmark & \cmark & \xmark \\
        \texttt{cmfrec}   & NMF & \cmark & \xmark & \xmark & \cmark & \xmark & \xmark & \xmark & \cmark & \cmark & \cmark & \Cpp & \xmark & \cmark & \xmark \\
        \texttt{gllvm}    & GMF & \cmark & \cmark & \cmark & \cmark & \cmark & \xmark & \xmark & \cmark & \cmark & \cmark & \R   & \xmark & \xmark & \xmark \\
        \texttt{GFM}      & GMF & \cmark & \xmark & \cmark & \cmark & \xmark & \xmark & \xmark & \cmark & \xmark & \cmark & \Cpp & \xmark & \xmark & \xmark \\
        \texttt{COAP}     & GMF & \xmark & \xmark & \xmark & \cmark & \xmark & \xmark & \xmark & \cmark & \xmark & \cmark & \Cpp & \xmark & \xmark & \xmark \\
        \texttt{glmpca}   & GMF & \xmark & \xmark & \cmark & \cmark & \cmark & \xmark & \xmark & \cmark & \cmark & \cmark & \R   & \xmark & \xmark & \cmark \\
        \texttt{zinbwave} & GMF & \xmark & \xmark & \xmark & \xmark & \cmark & \xmark & \cmark & \cmark & \cmark & \cmark & \R   & \xmark & \cmark & \xmark \\
        \texttt{newwave}  & GMF & \xmark & \xmark & \xmark & \xmark & \cmark & \xmark & \xmark & \cmark & \cmark & \cmark & \R   & \xmark & \cmark & \cmark \\
        \texttt{gmf}      & GMF & \cmark & \cmark & \cmark & \cmark & \xmark & \xmark & \xmark & \cmark & \xmark & \cmark & \R   & \cmark & \cmark & \xmark \\
        \texttt{dmf}      & GMF & \cmark & \cmark & \cmark & \cmark & \cmark & \cmark & \xmark & \cmark & \xmark & \cmark & \R   & \xmark & \xmark & \xmark \\
        \rowcolor{lightgray!35}
        \texttt{sgdGMF}   & GMF & \cmark & \cmark & \cmark & \cmark & \cmark & \cmark & \xmark & \cmark & \cmark & \cmark & \Cpp & \cmark & \cmark & \cmark \\
        \bottomrule
    \end{tabular}
\end{table}

Alongside, we propose an efficient \R/\Cpp~implementation of the proposed method in the new open-source \R~package \sgdGMF, freely available on CRAN\footnote{See \href{https://CRAN.R-project.org/package=sgdGMF}{\texttt{https://CRAN.R-project.org/package=sgdGMF}}}.
Compared to alternative implementations in \R, \sgdGMF~offers one of the most complete and flexible estimation frameworks for generalized matrix factorization modelling, allowing for all standard exponential family distributions, quasi-likelihood models, row- and column-specific regression effects, as well as model-based missing value imputation.
Moreover, it provides several algorithms for parameter estimation, including the proposed stochastic gradient descent approach.
To enhance scalability, parallel computing is employed as much as possible at any stage of the analysis, including initialization, estimation, model selection and post-processing.
Table \ref{tab:package_survey} compares all these features with alternative packages freely available in \R.

We showcase the \sgdGMF~implementation on both simulated and real data, demonstrating the scalability of the proposed method on gene expression matrices of different dimensions.
In all the numerical experiments we consider, the proposed stochastic gradient approach outperforms the alternative state-of-the-art methods in terms of execution time while having superior signal reconstruction quality, measured as out-of-sample residual deviance, logarithmic root mean squared error, and cluster separation in the latent space.

The paper is organized as follows. 
In Section \ref{sec:model_specification} we formally define the class of generalized matrix factorization models, we formulate the associated estimation problem and we discuss the connections with other models in the literature.
In Section \ref{sec:estimation_algorithms}, we introduce the proposed stochastic optimization method building upon quasi-Newton and stochastic gradient descent algorithms; additionally, we propose an efficient approach for parameter initialization.
In Section \ref{sec:additional_computational_aspects}, we briefly discuss some additional computational aspects, such as initialization and model selection.
In Section \ref{sec:simulation_studies}, we empirically compare the proposed algorithm with several state-of-the-art methods in the literature through an extensive simulation study.
In Section \ref{sec:real_data_applications}, we propose two case studies on real datasets of medium- and high-dimensional sizes, respectively, proving the effectiveness of our approach to extract coherent biological signals while maintaining a high level of computational efficiency. 
Section \ref{sec:discussion} is devoted to a concluding discussion and future research directions.

%%% SECTION
\section{Model specification}%
\label{sec:model_specification}

Let us define $\bY = \{ y_{ij} \}$ as the $n \times m$ data matrix containing the response variables of interest, which have $(i,j)$th entry $y_{ij} \in \cY \subseteq \real$, $i$th row $\by_{i:} \in \cY^m$ and $j$th column $\by_{:j} \in \cY^n$.
Conventionally, $y_{ij}$ is here considered as the $i$th observational unit over the $j$th variable, that in our genomic application are the $i$th cell and the $j$th feature/gene.
Hereafter, all the vectors are column vectors and all the transposed vectors, denoted by $\cdot^\top$, are row vectors.
To account for non-Gaussian observations, such as odds, counts, or continuous positive scores, we consider for the response variable $y_{ij}$ an \emph{exponential dispersion family} (EF) distribution with \textit{natural parameter} $\theta_{ij}$ and \textit{dispersion} $\phi$, denoted by
\begin{equation}
    \label{eq:exponential_family_distribution}
    y_{ij} \mid \theta_{ij} \sim \EF(\theta_{ij}, \phi), 
    \quad i = 1, \dots, n,
    \quad j = 1, \dots, m.
\end{equation}
Moreover, we denote the mean and variance of $y_{ij}$ by $\E(y_{ij}) = \mu_{ij}$ and $\Var(y_{ij}) = a_{ij}(\phi) \,\nu(\mu_{ij})$, respectively.
Here and elsewhere, $\nu(\cdot)$ is the family-specific variance function, which controls the heteroscedastic relationship between the mean and variance of $y_{ij}$, while $a_{ij}(\cdot)$ is a dispersion function specified as $a_{ij}(\phi) = \phi / w_{ij}$, with $w_{ij} > 0$ being a user-specified weight and $\phi > 0$ being a scalar dispersion parameter.
Under such a model specification, the probability density function of $y_{ij}$ can be written as
\begin{equation}
    \label{eq:exponential_family_density}
    f(y_{ij}; \theta_{ij}, \phi) = \exp \big[ \{ y_{ij} \theta_{ij} - b(\theta_{ij}) \} / a_{ij}(\phi) + c(y_{ij}, \phi) \big],
\end{equation}
where $b(\cdot)$ and $c(\cdot,\cdot)$ are family-specific functions. 
Specifically, $b(\cdot)$ is a convex twice differentiable \textit{cumulant function} and $c(\cdot,\cdot)$ is the so-called \textit{log-partition function}.
By the fundamental properties of the dispersion exponential family, the mean and variance of $y_{ij}$ satisfy the identities $\mu_{ij} = \dt{b}(\theta_{ij})$ and $\nu(\mu_{ij}) = \ddt{b}(\theta_{ij})$, which imply $\nu(\cdot) = (\ddt{b} \circ \dt{b}^{-1})(\cdot)$, where $\dt{b}(\cdot)$ and $\ddt{b}(\cdot)$ denote the first and second derivatives of $b(\cdot)$, respectively.
Therefore, $\dt{b}(\cdot)$ is a bijective map and $\theta_{ij}$ is uniquely determined by $\mu_{ij}$.
We recall that the deviance function of model \eqref{eq:exponential_family_density} is defined as $D (y, \mu) = - 2 \log\{ f_\phi (y,\mu) / f_\phi (y,y) \}$, where $f_\phi (y,\mu) = f(y; \theta, \phi)$ denotes the likelihood function relative to observation $y$ expressed as a function of mean $\mu$ and possibly depending on dispersion $\phi$.

Some relevant distributions belonging to the exponential family are, among others, the Gaussian, Inverse Gaussian, Gamma, Poisson, Binomial, and Negative Binomial laws (see Supplementary Table \ref{tab:exponential_family_distributions} for family-specific variance functions and the associated canonical link and deviance functions of these examples).
Of particular interest for omics applications are the Poisson and Negative Binomial distributions, as the gene expressions are typically obtained from technologies that yield count data (e.g., RNA sequencing or in-situ hybridization). However, formulating the model in the general form of the exponential family allows for a more general treatment and for the application to other omics technologies, such as mass spectrometry and microarray assays, which yield continuous readouts.

We complete the model specification by introducing an appropriate parametrization for the conditional mean $\mu_{ij}$.
In particular, we consider the generalized multivariate regression model
\begin{equation}
    \label{eq:linear_predictor_model}
    g(\mu_{ij}) 
    = \eta_{ij} 
    = (\bX \bB^\top + \bGamma \bZ^\top + \bU \bV^\top)_{ij}
    = \bx_{i:}^\top \bbeta_{:j} + \bgamma_{i:}^\top \bz_{:j} + \bu_{i:}^\top \bv_{:j},
\end{equation}
where $g(\cdot)$ is a continuously differentiable bijective link function and $\eta_{ij}$ is a linear predictor.
The latter is represented as an additive decomposition of three terms: a column-specific regression effect, $(\bX \bB^\top)_{ij} = \bx_{i:}^\top \bbeta_{:j}$, a row-specific regression effect, $(\bGamma \bZ^\top)_{ij} = \bgamma_{i:}^\top \bz_{:j}$, and a residual matrix factorization $(\bU \bV^\top)_{ij} = \bu_{i:}^\top \bv_{:j}$.
Specifically, $\bx_{i:} \in \real^p$ and $\bz_{:j} \in \real^q$ denote observed covariate vectors, $\bbeta_{:j} \in \real^p$ and $\bgamma_{i:} \in \real^q$ are unknown regression coefficient vectors, while $\bu_{i:} \in \real^d$ and $\bv_{:j} \in \real^d$ encode latent traits explaining the residual modes of variation in the data that are not captured by the regression effects.
Finally, we introduce the vector of unknown parameters in the model as $\bpsi = (\bbeta^\top, \bgamma^\top, \bu^\top, \bv^\top, \phi)^\top$, where the lower-case letters represent the flat vectorization of the corresponding matrix forms; for instance, $\bbeta = \vec(\bB) \in \real^{pm}$.

Specifically, in the scRNA-seq context, the regression term $\bx_{i:}^\top \bbeta_{:j}$ is often used to control for cell-specific technical confounders, such as batch effects and individual characteristics in multi-subject studies, with $\bx_{i:}$ containing indicator variables of the group assignment and group-specific attributes. 
Similarly, the second regression term, $\bgamma_{i:}^\top \bz_{:j}$, can account for gene-specific information contained in the covariate vector $\bz_{:j}$, such as GC-content or known functional interactions between genes. 
Notice that, even in the absence of external confounders, such regression effects allow for the inclusion of cell- and gene-specific intercepts, say $\eta_{ij} = \beta_{0j} + \gamma_{0i} + \bu_{i:}^\top \bv_{j:}$, which play the role of centering factors in the linear predictor metric. 
In scRNA-seq, the cell-specific intercept $\gamma_{0i}$ is of particular importance, as it acts as a scale factor that allows the model to be applied to raw counts rather than library size normalized data \citep{Risso2018}, which is critical since normalization can lead to distortions \citep{Townes2019}.

The latent traits $\bu_{i:}$ can be interpreted as \emph{meta-gene} variables representing the fundamental biological characteristics of the $i$th cell in a latent low-dimensional space that accounts for the covariates effects, with interpretation similar to usual PCA. For example, the meta-gene variables can be used for data visualization, cell clustering, and lineage reconstruction.
It is worth noting that, unlike formulation \eqref{eq:linear_predictor_model}, most factorization models proposed in the literature and implemented in standard computing environments do not support the inclusion of row effects $\bx_{i:}^\top \bbeta_{:j}$'s, column effects $\bgamma_{i:}^\top \bz_{:j}$'s, or both, resulting in latent representations that still capture variability of unwanted confounders, such as batch effects, or that do not include a normalization factor.
This limitation applies to \GMF~\citep{Kidzinski2022}, \DMF~\citep{Wang2023}, \GFM~\citep{Liu2023, Nie2024}, \COAP~\citep{Liu2024}, \NMF~\citep{Gaujoux2010}, and \NNLM~\citep{Lin2020}, as summarized in Table~\ref{tab:package_survey}.

Following the naming convention introduced by \cite{Kidzinski2022}, we refer to the model specification presented so far in \eqref{eq:exponential_family_distribution}--\eqref{eq:linear_predictor_model} as \emph{generalized matrix factorization} (GMF).
Alternative nomenclatures, such as \emph{exponential family principal component analysis} (EPCA; \citealp{Collins2001, Mohamed2008, Li2010}), \emph{generalized low-rank models} (GLRM; \citealp{Udell2016}), \emph{generalized linear latent variable model} (GLLVM; \citealp{Niku2017, Hui2017}), \emph{generalized principal component analysis} (glmPCA; \citealp{Townes2019}), \emph{deviance matrix factorization} (DMF; \citealp{Wang2023}), or \emph{generalized factor model} (GFM; \citealp{Liu2023, Liu2024, Nie2024}), can also be found in the literature.
Straightforward extensions of the GMF specification include pseudo-likelihood models, with $f(y_{ij}; \bpsi) = \exp\{ - L_\phi(y_{ij}, \eta_{ij}) \}$ being the negative exponential of a loss function $L_\phi(\cdot, \cdot) : \cY \times \real \rightarrow \real_+$.
Also, vector generalized estimating equations for overdispersed data are a particular case of this more general setup.

\subsection{Parameter identifiability}
\label{subsec:parameter_identifiability}

The multivariate generalized linear model in \eqref{eq:linear_predictor_model} is non-identifiable, allowing for column-space overlapping, and being invariant with respect to rotation, scaling, and sign-flip transformations of $\bU$ and $\bV$. 
Then, to enforce the uniqueness of the matrix decomposition, we need to impose additional identifiability constraints.
First, we need to impose the orthogonality of the parameters with respect to the covariate column space:
\begin{itemize}
    \item[(A)]\label{ass:covariate_space_orthogonality} $\bX$ and $\bZ$ are full-column rank matrices, moreover $\bX^\top \bGamma = \bzero$, $\bX^\top \bU = \bzero$, and $\bZ^\top \bV = \bzero$.
\end{itemize}
The first conditions ensure that $\bX^\top \bX$ and $\bZ^\top \bZ$ are non-singular, which is a standard requirement for regression models. 
Conditions $\bX^\top \bGamma = \bzero$ and $\bX^\top \bU = \bzero$ prevent $\bGamma$ and $\bU$ from spanning the same column space of $\bX$, and similarly $\bZ^\top \bV = \bzero$ prevents $\bV$ from spanning the same column space of $\bZ$ \citep[see, e.g.,][]{Liu2024}.

Then, we must ensure the identifiability of $\bU$ and $\bV$ with respect to rotation, scaling, and sign-flip.
To this end, some of the most common choices in the literature involve the following equivalent parameterizations:
\begin{itemize}
    \setlength\itemsep{-0.05cm}
    \item[(B1)]\label{ass:orthogonal_loadings} $\bU$ has orthogonal columns, $\bU^\top \bU = \bSigma$, $\bV$ has orthonormal columns, $\bV^\top \bV = \bI_d$, the first non-zero element of each column of $\bV$ is positive;
    \item[(B2)]\label{ass:orthogonal_scores} $\bU$ has orthonormal columns, $\bU^\top \bU = \bI_d$, $\bV$ has orthogonal columns, $\bV^\top \bV = \bSigma$, the first non-zero element of each column of $\bU$ is positive;
    \item[(B3)]\label{ass:standardized_scores} $\bU$ has standardized columns, say $\Var(\bU) = \frac{1}{n} \bU^\top (\bI_n - \frac{1}{n} \bone_n \bone_n^\top) \bU = \bI_d$, and $\bV$ is lower triangular with positive diagonal entries.
\end{itemize}
Here, $\bSigma = \diag(\sigma_1, \dots, \sigma_d)$ denotes the diagonal matrix collecting all the non-zero singular values of $\bU \bV^\top$ in decreasing order, $\bI_d$ denotes the $d \times d$ identity matrix, $\bzero_d$ is the $d \times 1$ vector of zeros, and $\bone_n$ is the $n \times 1$ vector of ones.
In Appendix~\ref{app:parameter_identifiability} in the Supplementary Material, we prove that constraints (A), together with one of (B1), (B2), or (B3), are sufficient to ensure the identifiability of the model parameters.
Notice that any unconstrained estimate of $\bB$, $\bGamma$, $\bU$ and $\bV$ can be easily projected into the constrained space induced by the identifiability restrictions via post-processing; see Appendix~\ref{app:parameter_identifiability} in the Supplementary Material, along with, e.g., \cite{Kidzinski2022}, \cite{Wang2023}, and \cite{Liu2024}.
The choice of the parametrization typically depends on the specific application and the desired interpretation of the score and loading matrices.
(B1) is the standard parametrization in principal component analysis, (B2) is the usual parametrization in spectral analysis of digital signals, and (B3) is the most common parametrization in the factor model literature.
In our numerical experiments (Sections~\ref{sec:simulation_studies} and~\ref{sec:real_data_applications}), we use parametrization (B1), which is conventional in the RNA-seq literature; see, e.g., \cite{Risso2018}, \cite{Townes2019}, and \cite{Ahlmann2023}.

\subsection{Penalized likelihood estimation}
\label{subsec:penalized_likelihood_estimation}

In the statistical literature, the variables $\bu_{i:}$ are called \emph{latent factors} and, typically, are assumed to follow a standard independent $d$-variate Gaussian distribution.
This representation provides a complete specification of the probabilistic mechanism that generated the samples $\by_{i:}$'s, which are conditionally independent given $\bu_{i:}$'s.
The marginal log-likelihood function induced by such a latent variable representation is given by
\begin{equation}
    \label{eq:marginal_likelihood_function}
    \sum_{i = 1}^{n} \log \int_{\real^d} \Bigg[ \prod_{j = 1}^{m} f(y_{ij} \mid \bu_{i:}) \Bigg] f(\bu_{i:}) \,\d\bu_{i:},
\end{equation}
where $f(y_{ij} \,|\, \bu_{i:}) = f(y_{ij}; \theta_{ij}, \phi)$ is the conditional distribution of $y_{ij}$ given $\bu_{i:}$, while $f(\bu_{i:}) = \exp\{ - \bu_{i:}^\top \bu_{i:} / 2 \} / (2 \pi)^{d/2}$ is the marginal probability density function of $\bu_{i:}$.

The unknown non-stochastic parameters can be estimated via maximum likelihood by optimizing \eqref{eq:marginal_likelihood_function}.
To this end, many numerical approaches have been proposed in the literature, such as Laplace approximation \citep{Huber2004, Bianconcini2012}, adaptive quadrature \citep{Cagnone2013}, expectation-maximization \citep{Sammel1997, Cappe2009}, variational approximation \citep{Hui2017}.
In practice, all these strategies perform very well in terms of accuracy but are extremely computationally expensive and do not scale well in high-dimensional problems.

An alternative approach is to treat the latent factors as if they were non-stochastic parameters and estimate them together with the other unknown coefficients.
\cite{Kidzinski2022} motivated this approach as a form of \emph{penalized quasi-likelihood} (PQL, \citealp{Breslow1993}), which is a standard tool in the estimation of \emph{generalized linear mixed models} (GLMM, \citealp{Lee2017}).
Formally, the PQL estimate of $\bpsi$, say $\hat{\bpsi}$, is the solution of
\begin{equation}
    \label{eq:maximum_likelihood_estimate}
    \hat\bpsi = \argmin_{\bpsi \in \Psi} \;\{ \ell_\lambda (\bpsi; \by) \},
\end{equation}
with $\Psi$ being the parameter space of $\bpsi$, which incorporate the identifiability constrains, and $\ell_\lambda (\bpsi; \by)$ denoting the penalized negative log-likelihood function. 
The latter is given by
\begin{equation}
    \label{eq:penalized_likelihood_function}
    \ell_\lambda (\bpsi; \by) = - \sum_{i = 1}^{n} \sum_{j = 1}^{m} \log f(y_{ij}; \theta_{ij}, \phi) + \frac{\lambda}{2} \| \bU \|_F^2 + \frac{\lambda}{2} \| \bV \|_F^2,
\end{equation}
where $\lambda > 0$ is a regularization parameter and $\| \cdot \|_F^2$ denotes the Frobenius norm.
Frobenius penalization is often introduced for numerical stability issues, but it is also intrinsically connected to the rank determination problem.
Indeed, for any $\lambda > 0$, the Frobenius penalty implicitly shrinks the singular values of the matrix factorization toward zero, encouraging compact low-rank representations of the signal.
This spectral penalization effect and its connection with the nuclear norm have been extensively discussed in, e.g., \cite{Witten2009}, \cite{Mazumder2010}, and \cite{Kidzinski2022}.

In a complete data scenario, the number of observed data entries in the response matrix is $nm$ and the total number of unknown parameters to be estimated is $pm + qn + d(n+m) + 1$. 
In partially observed data cases, the sum over $i \in \{1, \dots, n\}$ and $j \in \{ 1, \dots, m \}$ in \eqref{eq:penalized_likelihood_function}, can be easily replaced with a sum over $(i,j) \in \Omega$, where $\Omega \subseteq \{1, \dots, n \} \times \{1, \dots, m\}$ is the set collecting the index-position of all the observed data in the response matrix.

The optimization problem \eqref{eq:maximum_likelihood_estimate} is not jointly convex in $\bU$ and $\bV$. %, thus making it difficult to check the convergence to a global minimum.
However, objective function \eqref{eq:penalized_likelihood_function} is bi-convex, namely it is conditional convex in $\bU$ given $\bV$, and \emph{vice versa}.
This characteristic naturally encourages the development of iterative methods, which cycle over the alternated updates of $\bU$ and $\bV$ until convergence to a local stationary point.
Similar strategies are commonly used in matrix completion and recommendation systems for the estimation of high-dimensional matrix factorization models; see, e.g., \cite{Zou2006}, \cite{Koren2009} and \cite{Mazumder2010}.

\subsection{Related models}
\label{subsec:related_models_and}

The GMF specified in \eqref{eq:exponential_family_distribution}--\eqref{eq:linear_predictor_model} has strict connections and similarities to several models in the literature.
It extends \emph{vector generalized linear models} (\vglm, \citealp{Yee2015}), and hence also univariate generalized linear models \citep{McCullagh1989}, by introducing a second regression effect, $(\bGamma \bZ^\top)_{ij}$, and a latent matrix factorization, $(\bU \bV^\top)_{ij}$, in the linear predictor, to account for additional modes of variations and residual dependence structures.

Also, GMF directly generalizes \emph{principal component analysis} (PCA), which, by definition, is the solution of the minimization problem
\begin{equation*}
    \min_{\bu, \bv} \;\frac{1}{2} \sum_{i = 1}^{n} \sum_{j = 1}^{m} (y_{ij} - \bu_{i:}^\top \bv_{j:})^2
    \quad\text{subject to}\quad \bU^\top \bU = \bI_d, \quad \bV^\top \bV = \bDelta^2.
\end{equation*}
Then, in the GMF notation, PCA can be obtained by assuming a Gaussian distribution for the data matrix, an identity link function for the mean and no regression effects in the linear predictor, namely $(\bX \bB^\top)_{ij} = 0$ and $(\bGamma \bZ^\top)_{ij} = 0$.
Generalizations of PCA, such as Binary PCA \citep{Schein2003, Lee2010, Landgraf2015, Song2019} and Poisson PCA \citep{Kenney2021, Virta2023}, are also included in the GMF framework.

Close connections can also be drawn with \emph{non-negative matrix factorization} (NMF; \citealp{Wang2012}), which, in its more common formulations, searches for the best low-rank approximation $\bU \bV^\top$ of the data matrix $\bY$ by minimizing either the squared error loss or the Poisson deviance under non-negativity constraints for $\bU$ and $\bV$.
Formally, the NMF solution is defined as 
\begin{equation*}
    \min_{\bu, \bv} \;\sum_{i = 1}^{n} \sum_{j = 1}^{m} L(y_{ij}, \bu_{i:}^\top \bv_{j:})
    \quad\text{subject to}\quad \bU \geq 0, \quad \bV \geq 0,
\end{equation*}
with $L(\cdot, \cdot)$ being either the squared error loss, i.e. the Gaussian deviance, or the Kullback-Leibler loss, i.e. the Poisson deviance.
This representation clarifies that NMF can be written as a particular instance of GMF for non-negative data, where Gaussian/Poisson likelihood is used together with an identity link and non-negativity constraints.
Extensions of basic NMF have also been proposed to introduce external information through the inclusion of covariates; see, e.g., \emph{collective matrix factorization} and \emph{content-aware recommendation systems} \citep{Singh2008, Cortes2018}.

\section{Estimation algorithm}%
\label{sec:estimation_algorithms}

For the sake of exposition, throughout this section, we assume without loss of generality that $\bB = \bzero$ and $\bGamma = \bzero$.
% \cristiansout{or, equivalently, that $\bU_\star$ and $\bV_\star$ are partially known matrices incorporating the regression effects and the covariate matrices. In the following, we suppress the $\star$ subscript to lighten the notation.}
Moreover, we consider $\phi$ as a known fixed parameter, noting that, if unknown, it can be estimated iteration-by-iteration using either a method of moments or via maximum likelihood. 
More details about the general case $\bB \neq \bzero$ and $\bGamma \neq \bzero$ and the estimation of $\phi$ are provided in Appendix~\ref{app:algorithmic_details} in the Supplementary Material.
Finally, we introduce the matrices $\dt\bD = \{\dt{D}_{ij}\} = \{\partial \ell_\lambda / \partial \eta_{ij}\}$ and $\ddt\bD = \{\ddt{D}_{ij}\} = \{\partial^2 \ell_\lambda / \partial \eta_{ij}^2\}$ for the derivatives of the deviance with respect to the linear predictor, where
\begin{equation*}
    \dt{D}(y, \mu) = \frac{\partial \ell_\lambda}{\partial \eta} = \frac{w (y - \mu)}{\phi \,\nu(\mu) \,\dt{g}(\mu)}, \qquad
    \ddt{D}(y, \mu) = \frac{\partial^2 \ell_\lambda}{\partial \eta^2} = \frac{w \,\alpha(\mu)}{\phi \,\nu(\mu) \{ \dt{g}(\mu)\}^2},
\end{equation*}
and $\alpha(\mu) = 1 + (y - \mu) \{ \dt\nu(\mu) / \nu(\mu) + \ddt{g}(\mu) / \dt{g}(\mu) \}$.
Accordingly, we define $\dt{g} = \d{g} / \d\mu$, $\ddt{g} = \d^2{g} / \d\mu^2$ and $\dt\nu = \d\nu / \d\mu$.
The expected second-order derivative, i.e., the Fisher weight $\E (\ddt{D}_{ij})$, just corresponds to the observed second-order differential $\ddt{D}_{ij}$ with $\alpha(\mu_{ij}) = 1$, which is positive for any $y_{ij}$ and $\mu_{ij}$.

The penalized estimate $\hat\bpsi$ in \eqref{eq:maximum_likelihood_estimate}, say the vectorized concatenation of $\hat\bU$ and $\hat\bV$, must satisfy the first-order matrix conditions
\begin{equation}
    \label{eq:first_order_equations}
    \frac{\partial \ell_\lambda}{\partial \bU} = \dt{\bD} \,\bV + \lambda \bU = 0, \qquad
    \frac{\partial \ell_\lambda}{\partial \bV} = \dt{\bD}^\top \bU + \lambda \bV = 0,
\end{equation}
where the differentiation is performed element-wise. % say $\partial \ell_\lambda / \partial \bU = \{ \partial \ell_\lambda / \partial u_{ih}\}$ and $\partial \ell_\lambda / \partial \bV = \{ \partial \ell_\lambda / \partial v_{jh}\}$. 
Using the same formulation, we may also express the second-order derivatives of \eqref{eq:penalized_likelihood_function} as
\begin{equation}
    \label{eq:second_order_derivatives}
    \frac{\partial^2 \ell_\lambda}{\partial \bU^2} = \ddt{\bD} \,(\bV * \bV) + \bLambda > 0, \qquad
    \frac{\partial^2 \ell_\lambda}{\partial \bV^2} = \ddt{\bD}^\top (\bU * \bU) + \bLambda > 0,
\end{equation}
where $*$ is the Hadamard product and $\bLambda$ is a matrix of appropriate dimensions filled by $\lambda$.

Conditionally on $\bV$, the left matrix equation in \eqref{eq:first_order_equations} can be decomposed into $n$ multivariate equations  row-by-row, $\partial \ell_\lambda / \partial \bu_{i:} = 0$ ($i = 1, \dots, n$), that can be solved independently in parallel.
The existence and uniqueness of the solution of each row-equation are guaranteed under mild regularity conditions on the exponential family and the link function.
In the same way, the right matrix equation in \eqref{eq:first_order_equations} can be split into $m$ independent vector equations, $\partial \ell_\lambda / \partial \bv_{j:} = 0$ ($j = 1, \dots, m$), to be solved in parallel.
See, e.g.,  \cite{Kidzinski2022} and \cite{Wang2023} for a detailed discussion and derivation on \eqref{eq:first_order_equations} and \eqref{eq:second_order_derivatives}.
Notice that, in the presence of covariate effects, the derivatives in \eqref{eq:first_order_equations} must be replaced by $\partial \ell_\lambda / \partial [\bGamma, \,\bU]$ and $\partial \ell_\lambda / \partial [\bB, \,\bV]$, as detailed in Appendix~\ref{app:algorithmic_details} in the Supplementary Material.

%%% SUBSECTION
\subsection{Fisher scoring and quasi-Newton algorithms}
\label{subsec:fisher_scoring_and_quasi_newton}

The first, and most popular, algorithm introduced in the literature for finding the solution of \eqref{eq:first_order_equations} is the \emph{alternated iterative re-weighted least squares} (AIRWLS, \citealp{Collins2001, Li2010, Risso2018, Kidzinski2022, Wang2023, Liu2023}) method.
It cycles between the conditional updates of $\bU$ and $\bV$ by solving the equations in \eqref{eq:first_order_equations} in a row-wise manner, using standard Fisher scoring for GLMs \citep{McCullagh1989}.
The resulting routine is statistically motivated, easy to implement and allows for efficient parallel computing.

However, in massive data settings, it becomes infeasible when the dimension of the problem increases in the sample size or the latent space rank. 
One iteration of the algorithm, i.e., a complete update of $\bU$ and $\bV$, needs $O((n+m) d^3 + nmd)$ floating point operations to be performed, where the leading term proportional to $d^3$ comes from a matrix inversion that must be computed at least $n+m$ times per iteration.
This is particularly limiting in real-data applications, since $d$ is unknown \emph{a priori} and must be selected in a data driven way, which might require fitting the model several times and for an increasing number of $d$, which scales cubically.

To overcome this issue, \cite{Kidzinski2022} proposed a quasi-Newton algorithm which employs an approximate inversion only using the diagonal elements of the Fisher information matrix.
With this simplification, only elementary matrix operations are performed reducing the computational complexity to $O((n+m)d + nmd)$. 
In formulas, the quasi-Newton algorithm of \cite{Kidzinski2022} updates the parameter estimates at iteration $t$ as
\begin{equation}
    \label{eq:quasi_newton_update_psi}
    \bpsi^{t+1} \gets \bpsi^t + \rho_t \,\bDelta_\psi^t, \qquad 
    \bDelta_\psi^t = - (\bG_\psi^t \,/ \,\bH_\psi^t),
\end{equation}
where $\{ \rho_t \}$ is a sequence of learning rate parameters, $\bDelta_\psi^t$ is the search direction, while $\bG_\psi^t = \partial \ell_\lambda^t / \partial \bpsi$ and $\bH_\psi^t = \partial^2 \ell_\lambda^t / \partial \bpsi^2$ denote respectively the first two derivatives of $\ell_\lambda(\bpsi; \by)$ with respect to $\bpsi$ evaluated at $\bpsi^t$.
Throughout, $\gets$ stands for the assignment operator and the division is performed element-wise.
Exploiting the block structure of $\bpsi$, the joint update \eqref{eq:quasi_newton_update_psi} can be written in the coordinate-wise form
\begin{equation}
    \label{eq:quasi_newton_update}
    \begin{aligned}
        \bU^{t+1} \gets \bU^t + \rho_t \bDelta_\scu^t, \qquad & \bDelta_\scu^t = - (\bG_\scu^t \,/ \,\bH_\scu^t), \\
        \bV^{t+1} \gets \bV^t + \rho_t \bDelta_\scv^t, \qquad & \bDelta_\scv^t = - (\bG_\scv^t \,/ \,\bH_\scv^t).
    \end{aligned}
\end{equation}
where $\bG_\scu$, $\bG_\scv$, $\bH_\scu$ and $\bH_\scv$ can be obtained as in \eqref{eq:first_order_equations} and \eqref{eq:second_order_derivatives} \citep{Kidzinski2022}.
Overall, each quasi-Newton iteration requires $O((n+m)d + n m d)$ floating point operations and $O((n+m)d + n m)$ memory allocations.
To the best of our knowledge, this is the most efficient algorithm in the literature for the estimation of GMF models.

In what follows, we build upon the quasi-Newton algorithm of \cite{Kidzinski2022} to derive an efficient stochastic optimization method that, for our purposes, should further improve the scalability of GMF modelling in high-dimensional settings.

%%% SUBSECTION
\subsection{Stochastic gradient descent}
\label{subsec:stochastic_gradient_descent}

Stochastic gradient descent (SGD, \citealp{Bottou2010}) provides an easy and effective strategy to handle complex optimization problems in massive data applications.
Similarly to deterministic gradient-based methods, SGD~is an iterative optimization procedure which updates the parameter vector $\bpsi$ until convergence following the approximate steepest descent direction, say $\bDelta_\psi^t = - \hat\bG_\psi^t$.
Here, the hat-notation, $\hat\bG_\psi^t$, stands for an unbiased stochastic estimate of $\bG_\psi^t$.
Under mild regularity conditions on the optimization problem and the learning rate sequence, specifically $\sum_{t = 0}^{\infty} \rho_t = \infty$ and $\sum_{t = 0}^{\infty} \rho_t^2 < \infty$, SGD~is guaranteed to converge to a stationary point of the objective function \citep{Robbins1951}.
A standard choice for the learning rate sequence is $\rho_t = k_0 / (1 + k_0 k_1 t)^{\tau}$ for $k_0, k_1 > 0$ and $\tau \in (1/2, 1]$, where $k_0$ is the initial stepsize, $k_1$ is a decay rate parameter, and $\tau$ determines the asymptotic tail behaviour of the sequence for $t \rightarrow \infty$. 

\subsubsection{Improving na{\"i}ve stochastic gradient}
\label{subsubsec:improving_naieve_stochastic_gradient}

In the past two decades, an active area of research built upon na\"{i}ve SGD~to improve its convergence speed, robustify the search path against erratic perturbations in the gradient estimate and introduce locally adaptive learning rate schedules using approximate second-order information.
Some examples are Nesterov acceleration~\citep{Nesterov1983}, \sgdqn~\citep{Bordes2009}, \adagrad~\citep{Duchi2011}, \adadelta~\citep{Zeiler2012}, \rmsprop~\citep{Hinton2012}, \adam~\citep{Kingma2014}, \amsgrad~\citep{Reddi2019}.
Inspired by this line of literature, we propose the following \emph{adaptive stochastic gradient descent} (\asgd) updating rule
\begin{equation}
    \label{eq:stochastic_quasi_newton_update_psi}
    \bpsi^{t+1} \gets \bpsi^t + \rho_t \,\bDelta_\psi^t, \qquad 
    \bDelta_\psi^t = - \alpha_t \,(\bar\bG_\psi^t \,/ \,\bar\bH_\psi^t),
\end{equation}
where $\bar\bG_\psi^t$ and $\bar\bH_\psi^t$ are smoothed estimates of the derivatives of \eqref{eq:penalized_likelihood_function}, while $\alpha_t$ is a scalar bias-correction factor.
Similar to \adam~\citep{Kingma2014} and \amsgrad~\citep{Reddi2019}, we update the smoothed gradients, $\bar\bG_\psi^t$ and $\bar\bH_\psi^t$, using an exponential moving average of the previous gradient values and the current stochastic estimates $\hat\bG_\psi^t$ and $\hat\bH_\psi^t$, namely
\begin{equation}
    \label{eq:gradient_exponential_averaging}
    \bar\bG_\psi^t \gets (1 - \alpha_1) \,\bar\bG_\psi^{t-1} + \alpha_1 \,\hat\bG_\psi^t, \qquad
    \bar\bH_\psi^t \gets (1 - \alpha_2) \,\bar\bH_\psi^{t-1} + \alpha_2 \,\hat\bH_\psi^t,
\end{equation}
with $\alpha_1, \alpha_2 \in (0,1]$ being user-specified smoothing coefficients, while $\alpha_t = (1 - \alpha_2^t) / (1 - \alpha_1^t)$ is introduced to filter out the bias induced by the exponential moving average in \eqref{eq:gradient_exponential_averaging}. 
Typical values for the smoothing coefficients are $\alpha_1 = 0.1$ and $\alpha_2 = 0.01$, respectively.
Such a smoothing technique has two main advantages: it speeds up the convergence using the inertia accumulated by previous gradients and, at the same time, it stabilizes the optimization, reducing the noise around the gradient estimate.
We refer the reader to \cite{Kingma2014} and \cite{Reddi2019} for a deeper discussion on the benefits of bias corrected exponential gradient averaging in high-dimensional stochastic optimization problems.

Different from the original implementation of \adam~\citep{Kingma2014}, which only involves first-order information, we scale the smoothed gradient estimate using a smoothed diagonal Hessian approximation.
This strategy allows us to directly extend the diagonal quasi-Newton algorithm of \cite{Kidzinski2022} in a stochastic vein without increasing the computational overload or the stability of the optimization. 
Indeed, the derivatives of the deviance function of a \glm~model are always well-defined, bounded away from zero, available in closed form and easy to compute.
This is not the case for deep neural networks and other machine learning models for which \adam~and \amsgrad~have been originally developed. 

Analogously to the quasi-Newton update \eqref{eq:quasi_newton_update} for GMF models, also \eqref{eq:stochastic_quasi_newton_update_psi} can be written in the block-wise form 
\begin{equation}
    \label{eq:stochastic_quasi_newton_update}
    \begin{aligned}
        \bU^{t+1} \gets \bU^t + \rho_t \,\bDelta_\scu^t, \qquad & \bDelta_\scu^t = - \alpha_t (\bar\bG_\scu^t / \bar\bH_\scu^t), \\
        \bV^{t+1} \gets \bV^t + \rho_t \,\bDelta_\scv^t, \qquad & \bDelta_\scv^t = - \alpha_t (\bar\bG_\scv^t / \bar\bH_\scv^t).
    \end{aligned}
\end{equation}
In this formulation, update \eqref{eq:quasi_newton_update} and \eqref{eq:stochastic_quasi_newton_update} have a computational complexity proportional to the dimension of the matrices $\bU$ and $\bV$, namely $O((n+m) d)$.
Since the dimension of the parameter space can not be further compressed without losing prediction power, the alternative way to speed up the computations is to efficiently approximate the gradients, $\hat\bG_\psi^t = (\hat\bG_\scu^t, \hat\bG_\scv^t)$ and $\hat\bH_\psi^t = (\hat\bH_\scu^t, \hat\bH_\scv^t)$.

Exploiting the additive structure of the penalized log-likelihood in \eqref{eq:penalized_likelihood_function}, a natural unbiased estimate of its derivatives can be obtained via sub-sampling.
In particular, we can estimate the log-likelihood gradient by summing up only a small fraction of the data contributions, which we refer to as a \emph{minibatch}, instead of the whole dataset.
This strategy is highly scalable and can be calibrated on the available computational resources.

% In regression and classification modelling, the log-likelihood function can be represented as a sum of independent terms and its gradient shares the same additive structure.
% Therefore, a natural estimate of the log-likelihood, as well as its derivatives, can be obtained by summing up only a small fraction of the data contributions to the likelihood, say a \emph{minibatch}, instead of the whole dataset.
% This strategy is highly scalable and can be calibrated on the available computational resources.

In what follows, we discuss a new stochastic optimization method for GMF models based on formula \eqref{eq:stochastic_quasi_newton_update}, which uses local parameter updates along with minibatch estimates of the current gradients to reduce the computational complexity of the resulting algorithm.
To this end, we introduce the following notation: $I = \{ i_1, \dots, i_{\nrow} \} \subseteq \{1, \dots, n\}$ is a subset of row-indices of dimension $\nrow$, $J = \{ j_1, \dots, j_{\ncol} \} \subseteq \{ 1, \dots, m \}$ is a subset of column-indices of dimension $\ncol$, $B = I \times J$ is the Cartesian product between $I$ and $J$. 
Finally, $\bY_{\row :} = \{ \by_{i:} \}_{i \in I}$, $\bY_{: \col} = \{ \by_{:j} \}_{j \in J}$ and $\bY_{\mb} = \{ y_{ij} \}_{(i,j) \in B}$ denote the corresponding sub-matrices, also called \emph{row-}, \emph{column-} and \emph{block-minibatch} subsamples of the original data matrix, respectively.

%%% SUBSECTION
\subsubsection{Block-wise adaptive stochastic gradient descent}
\label{subsubsec:blockwise_stochastic_gradient_descent}

% Let $\cI$ and $\cJ$ be, respectively, a partition of the row-index set $\{ 1, \dots, n \}$ and a partition of the column-index set $\{ 1, \dots, m\}$ such that $I \in \cI$ and $J \in \cJ$

Both the exact and stochastic quasi-Newton methods identified by equations \eqref{eq:quasi_newton_update} and \eqref{eq:stochastic_quasi_newton_update} entirely update $\bU$ and $\bV$ at each iteration.
Despite being an effective and parallelizable strategy, in many applicative contexts, when it comes to factorizing huge matrices, an entire update of the parameters could be extremely expensive in terms of memory allocation and execution time.
This is a well-understood issue in the literature on recommendation systems, where standard matrix completion problems may involve matrices with millions of rows and columns; see, e.g., \cite{Koren2009}, \cite{Mairal2010}, \cite{Recht2013}, \cite{Mensch2017}.
Moreover, batch optimization strategies do not generalize well to stream data contexts, where the data arrive sequentially and the parameters must be updated on-the-fly as new sets of observations come in.

A classic solution is to perform iterative element-wise SGD~steps using only one entry of the data matrix, $y_{ij}$, at each iteration, thus updating the low-rank decomposition matrices row-by-row through the paired equations $\bu_{i:}^{t+1} \gets \bu_{i:}^t - \rho_t \bG_{\scu, i:}^t$ and $\bv_{j:}^{t+1} \gets \bv_{j:}^t - \rho_t \bG_{\scv, j:}^t$.
In the matrix factorization literature, this strategy is also known as \emph{online SGD}~algorithm, since it permits updating the parameter estimates dynamically when a new observation is gathered.

To stabilize the optimization and speed up the convergence, we generalize the online SGD~approach in two directions: we consider local stochastic quasi-Newton updates in place of na\"{i}ve gradient steps, and we use block-wise minibatches of the original data matrix, $\bY_\mb$, instead of singletons, $y_{ij}$.
In formulas, the algorithm we propose cycles over the following adaptive gradient steps:
\begin{equation}
    \label{eq:blockwise_stochastic_quasi_newton_update}
    \begin{aligned}
        \bU_{\row :}^{t+1} \gets \bU_{\row :}^t + \rho_t \,\bDelta_{\scu, \row :}^t, \qquad & 
        \bDelta_{\scu, \row :}^t = - \alpha_t (\bar\bG_{\scu, \row :}^t / \bar\bH_{\scu, \row :}^t), \\
        \bV_{\col :}^{t+1} \gets \bV_{\col :}^t + \rho_t \,\bDelta_{\scv, \col :}^t, \qquad & 
        \bDelta_{\scv, \col :}^t = - \alpha_t (\bar\bG_{\scv, \col :}^t / \bar\bH_{\scv, \col :}^t).
    \end{aligned}
\end{equation}
The smoothed gradients are then estimated by exponential average as in \eqref{eq:gradient_exponential_averaging} and the minibatch stochastic gradients are obtained as
\begin{equation}
    \label{eq:blockwise_minibatch_gradients}
    \begin{aligned}
        \hat\bG_{\scu,\row :}^t = (m / \ncol) \,\dt\bD_{\mb}^t \bV_{\col :}^t + \lambda \bU_{\row :}^t, \qquad &
        \hat\bH_{\scu,\row :}^t = (m / \ncol) \,\ddt\bD_{\mb}^t (\bV_{\col :}^t * \bV_{\col :}^t) + \bLambda, \\
        \hat\bG_{\scv,\col :}^t = (n / \nrow) \,\dt\bD_{\mb}^{t \top} \bU_{\row :}^t + \lambda \bV_{\col :}^t, \qquad &
        \hat\bH_{\scv,\col :}^t = (n / \nrow) \,\ddt\bD_{\mb}^{t \top} (\bU_{\row :}^t * \bU_{\row :}^t) + \bLambda.
    \end{aligned}
\end{equation}
Here, $\nrow$ and $\ncol$ denote the number of rows and columns of each minibatch and $(m / \ncol) \,(\dt\bD_{\mb}^t \bV_{\col :}^t)_{ih} = (m / \ncol) \,\sum_{j \in J}^{} \dt{D}_{ij}^t v_{jh}^t$ is an unbiased stochastic estimate of $(\dt\bD_{}^t \bV)_{ih} = \sum_{j = 1}^{m} \dt{D}_{ij}^t v_{jh}^t$ for $i \in I$. 
Similarly, it is easy to show that the other minibatch averages in \eqref{eq:blockwise_minibatch_gradients} are unbiased estimates of the corresponding batch quantities.
Figure \ref{fig:minibatch_subsampling} provides a graphical representation of the updates, while Algorithm \ref{alg:block_sgd_algorithm} provides a pseudo-code description of the proposed procedure.
Overall, each iteration of Algorithm \ref{alg:block_sgd_algorithm} requires $O((\nrow + \ncol) d + \nrow \ncol d)$ floating point operations and $O((\nrow + \ncol) d + \nrow \ncol)$ memory allocations.

\begin{figure}[H]
    \centering
    \includegraphics[width = \textwidth]{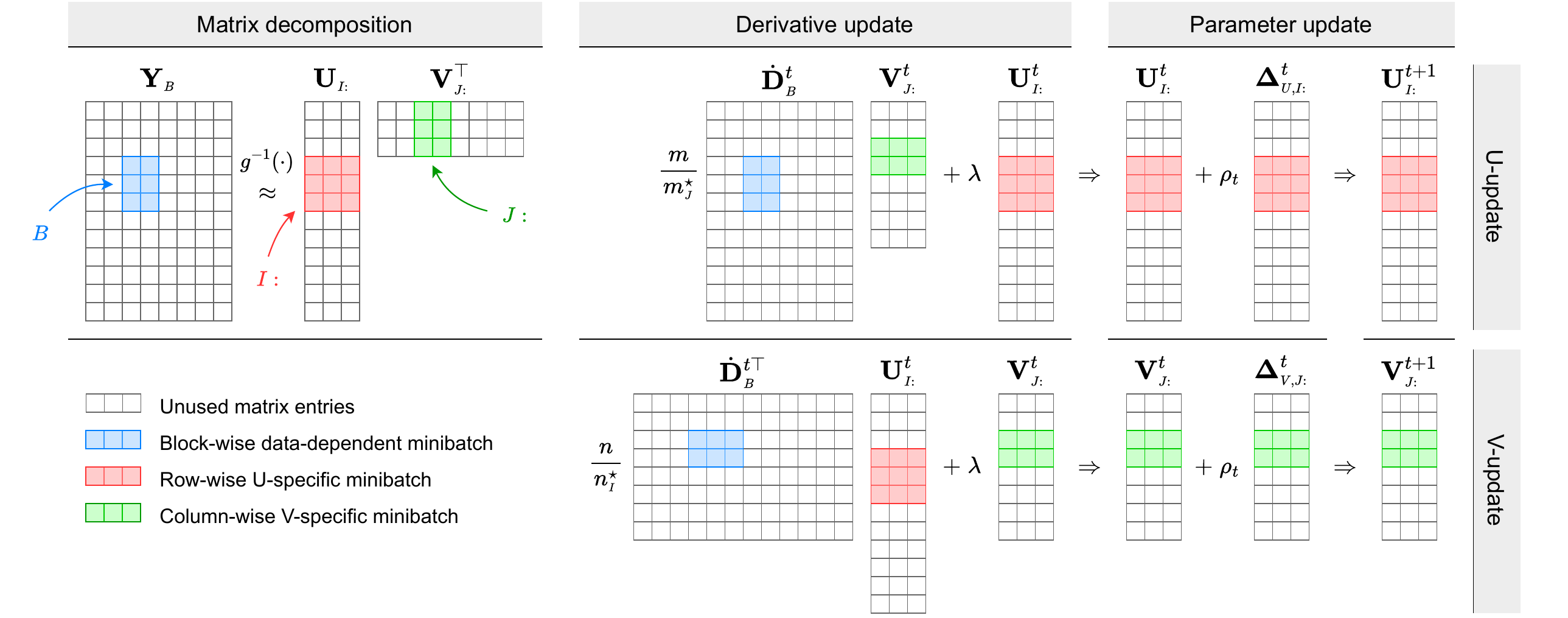}
    \caption{Graphical representation of the stochastic gradient updates employed at the $t$th iteration of Algorithm \ref{alg:block_sgd_algorithm}.
    Left: generalized matrix factorization model \eqref{eq:linear_predictor_model}.
    Middle: updates of the penalized log-likelihood gradients \eqref{eq:blockwise_minibatch_gradients}.
    Right: adaptive gradient step \eqref{eq:blockwise_stochastic_quasi_newton_update}. 
    The colored and empty cells highlight the sub-sampled data used and not used at the $t$th update, respectively.
    To save space, the calculations of the second-order differentials and the gradient smoothing are not displayed here.}
    \label{fig:minibatch_subsampling}
\end{figure}

\begin{algorithm}[t]
    \caption{%
        \label{alg:block_sgd_algorithm}
        Pseudo-code description of the block-wise adaptive SGD algorithm described in Section \ref{subsubsec:blockwise_stochastic_gradient_descent}.
        On the right, we report the computational complexity of each step.
    }
    
    Initialize $\bU$, $\bV$, $\bet$, $\bmu$, $\phi$\;
    Sample a random partition $\{ I_1, \dots, I_{\scr} \}$ such that $\cup_{r = 1}^{\,\scr} I_r = \{ 1, \dots, n \}$\;
    Sample a random partition $\{ J_1, \dots, J_{\scs} \}$ such that $\cup_{s = 1}^{\,\scs} J_s = \{1, \dots, m \}$\;
    
    \While{convergence is not reached}{

        \For{$J \in \{ J_1, \dots, J_{\scs} \}$}{
            {1. Sample the minibatch set} \\
            {
            \qquad Sample $I \in \{ I_1, \dots, I_{\scr} \}$ and set $B \gets I \times J$; \hfill
            $O(1)$ \\
            }
    
            %%% Sufficient statistics
            {2. Compute the subsampled likelihood derivatives} \\
            {
            \qquad $\displaystyle \bet_\mb^t \gets \bU_{\row :}^t \bV_{\col :}^{t \top}$;  \quad
            $\displaystyle \bmu_\mb^t \gets g^{-1}(\bet_\mb^t)$; \hfill
            $O(\nrow \ncol d)$ \\

            \qquad $\displaystyle \dt{\bD}_\mb^t \gets \bW_\mb * (\bY_\mb - \bmu_\mb^t) /  \big\{ \phi^t \,\nu(\bmu_\mb^t) * \dt{g}(\bmu_\mb^t) \big\}$; \hfill
            $O(\nrow \ncol)$ \\
            
            \qquad $\displaystyle \ddt{\bD}_\mb^t \gets \bW_\mb / \big\{ \phi^t \,\nu(\bmu_\mb^t) * \dt{g}(\bmu_\mb^t)^2 \big\}$; \hfill
            $O(\nrow \ncol)$ \\
            }
    
            %%% V update
            {3. Compute the smoothed gradients and update $\bV$} \\
            {
            \qquad $\displaystyle \hat{\bG}_{\scv,\col :}^t \gets (n / n_\row^*)   \dt{\bD}_\mb^{t\top} \bU_{\row :}^t - \lambda \bV_{\col :}^t$; \quad
            $\displaystyle \hat{\bH}_{\scv,\col :}^t \gets (n / n_\row^*) \ddt{\bD}_\mb^{t\top} (\bU_{\row :}^t * \bU_{\row :}^t) - \bLambda$;  \hfill
            $O(\nrow \ncol d)$ \\
    
            \qquad $\bar{\bG}_{\scv,\col :}^t \gets (1 - \alpha_1) \bar{\bG}_{\scv,\col    :}^{t-1} + \alpha_1 \hat{\bG}_{\scv,\col :}^t$; \quad
            $\bar{\bH}_{\scv,\col :}^t \gets (1 - \alpha_2) \bar{\bH}_{\scv,\col    :}^{t-1} + \alpha_2 \hat{\bH}_{\scv,\col :}^t$; \hfill
            $O(\ncol d)$ \\
    
            \qquad $\displaystyle \bDelta_{\scv, \col :}^t \gets - \alpha_t \, (\bar{\bG}_{\scv,\col :}^t / \bar{\bH}_{\scv,\col :}^t)$; \quad
            $\displaystyle \bV_{\col :}^{t+1} \gets \bV_{\col :}^t + \rho_t     \bDelta_{\scv, \col :}^t$; \hfill
            $O(\ncol d)$ \\
            }
            
            %%% U update
            {4. Compute the smoothed gradients and update $\bU$} \\
            {
            \qquad $\displaystyle \hat{\bG}_{\scu,\row :}^t \gets (m / m_\col^*)   \dt{\bD}_\mb^t \,\bV_{\col :}^t - \lambda \bU_{\row :}^t$; \quad
            $\displaystyle \hat{\bH}_{\scu,\row :}^t \gets (m / m_\col^*)   \ddt{\bD}_\mb^t \,(\bV_{\col :}^t * \bV_{\col :}^t) - \bLambda$; \hfill
            $O(\nrow \ncol d)$ \\
    
            \qquad $\bar{\bG}_{\scu,\row :}^t \gets (1 - \alpha_1) \bar{\bG}_{\scu,\row    :}^{t-1} + \alpha_1 \hat{\bG}_{\scu,\row :}^t$; \quad
            $\bar{\bH}_{\scu,\row :}^t \gets (1 - \alpha_2) \bar{\bH}_{\scu,\row    :}^{t-1} + \alpha_2 \hat{\bH}_{\scu,\row :}^t$; \hfill
            $O(\nrow d)$ \\ 
    
            \qquad $\displaystyle \bDelta_{\scu, \row :}^t \gets - \alpha_t \, (\bar{\bG}_{\scu,\row :}^t / \bar{\bH}_{\scu,\row :}^t)$; \quad
            $\displaystyle \bU_{\row :}^{t+1} \gets \bU_{\row :}^t + \rho_t     \,\bDelta_{\scu, \row :}^t$; \hfill
            $O(\nrow d)$ \\
            }
        }
    }

    5. Orthogonalize $\hat\bU$ and $\hat\bV$;
\end{algorithm}

If the dispersion parameter $\phi$ is unknown and has to be estimated from the data, we can also adopt a smoothed stochastic estimator obtained through exponential averaging. More details are provided in Appendix~\ref{app:algorithmic_details} in the Supplementary Material.

Whenever the input data matrix is only partially complete, it is necessary to properly handle the missing values during the estimation process.
To this end, we rely on the general framework proposed by \cite{Cai2010} and \cite{Mazumder2010}, and later used by \cite{Kidzinski2022}.
This prescribes imputing the missing data entries during the optimization by updating them at each iteration using the most recent prediction of those values.
Algorithm \ref{alg:block_sgd_algorithm} can thus be adapted by replacing the static (incomplete) minibatch matrix $\bY_\mb$ with its completed version $\bY_\mb^t$ obtained at iteration $t$ after imputation. 
At the beginning of each iteration, we then introduce the imputation step as $y_{ij}^{t} \gets \mu_{ij}^{t}$ if $(i,j) \in B \cap \Omega^c$, where $\Omega^c$ is the complement of $\Omega$.

We refer to \eqref{eq:blockwise_stochastic_quasi_newton_update} as \emph{local} or \emph{partial} updates since they just modify the rows of $\bU$ and $\bV$ corresponding to the minibatch block of indices $B = I \times J$.
On the contrary, we refer to updates \eqref{eq:quasi_newton_update} and \eqref{eq:stochastic_quasi_newton_update} as \emph{global} updating rules.
The choice between alternated least squares \citep{Kidzinski2022, Wang2023}, exact quasi-Newton \citep{Kidzinski2022} and the proposed adaptive stochastic gradient descent (Algorithm \ref{alg:block_sgd_algorithm}, Figure \ref{fig:minibatch_subsampling}) methods is up to the researcher and, in principle, depends on the dimension of the problem under study and on the available computational resources.
In general, algorithms using a higher amount of information are more stable and accurate; however, they tend to scale poorly in high-dimensional settings and often get stuck on suboptimal optima.
On the other hand, cheaper stochastic methods using less information scale well in big-data problems at the cost of a lower level of precision.

\section{Additional computational aspects}
\label{sec:additional_computational_aspects}

\subsection{Parameter initialization}%
\label{subsec:parameter_initialization}

As is common in non-convex optimization, the performance of Algorithm~\ref{alg:block_sgd_algorithm} heavily depends on its initialization.
Random starting values can slow convergence and increase the risk of getting stuck in poor local minima or unstable saddle points. To improve both performance and accuracy, we adopt a structured initialization strategy that leverages the conditional GLM formulation of model~\eqref{eq:linear_predictor_model}, and is inspired by the initialization approaches employed in \cite{Risso2018}, \cite{Townes2019}, and \cite{Kidzinski2022}.

We first estimate the column-specific regression parameters, $\bbeta_{j:}$'s, by fitting $m$ separate GLMs. 
Then, conditionally on the estimated offsets $\bx_{i:}^\top \tilde\bbeta_{j:}$, we estimate the row-specific parameters, $\bgamma_{i:}$'s, using the same strategy, and fitting $n$ separate GLMs. 
The resulting estimates, $\tilde\bbeta_{j:}$ and $\tilde\bgamma_{i:}$, initialize the regression effects.
To initialize the latent scores, $\bu_{i:}$'s, we adopt the null-residual method of \cite{Townes2019}, extracting the first $d$ left eigenvectors of a residual matrix based on either deviance or Pearson residuals: $r_{ij}^\textsc{d} = \sign(y_{ij} - \mu_{ij}) \sqrt{D(y_{ij}, \mu_{ij})}$ and $r_{ij}^\textsc{p} = (y_{ij} - \mu_{ij}) / \sqrt{\nu(\mu_{ij})}$, where the mean $\mu_{ij}$ is approximated as $\tilde\mu_{ij} = g^{-1}(\bx_{i:}^\top \tilde\bbeta_{j:} + \tilde\bgamma_{i:}^\top \bz_{j:})$.
Finally, we estimate the loadings, $\bv_{j:}$'s, by fitting $m$ column-specific GLMs with the latent scores as fixed design matrix and an offset term $\bx_{i:}^\top \tilde\bbeta_{j:} + \tilde\bgamma_{i:}^\top \bz_{j:}$ in the predictor.
Using standard solvers for GLMs and singular value decomposition, the overall computational cost is
\begin{align*}
    m \,\underbrace{O(p^3 + p^2 n)}_{\text{$\tB$--inner cycle}} + \;
    n \,\underbrace{O(q^3 + q^2 m)}_{\text{$\Gamma$--inner cycle}} + \;
    \underbrace{O((p+q+d) nm)}_{\text{$\tU$--residual SVD}} + \;
    m \,\underbrace{O(d^3 + d^2n)}_{\text{$\tV$--inner cycle}}.
\end{align*}
All GLM-fitting steps are highly parallelizable, with per-task complexity at most $O(k^3 + k^2 n)$, where $k = \max(p, q, d)$.
In high-dimensional settings, one may approximate the GLM fits via ordinary least squares on a link-transformed response, reducing complexity to:
\begin{align*}
    \underbrace{O(p^3 + p^2 n + pnm)}_{\text{$\tB$--least squares}} + \;
    \underbrace{O(q^3 + q^2 m + qnm)}_{\text{$\Gamma$--least squares}} + \;
    \underbrace{O((p+q+d) nm)}_{\text{$\tU,\tV$--residual SVD}}.
\end{align*}

Although inspired by the null-residual strategy of \cite{Townes2019}, our initialization differs in four key aspects:
(i) it serves solely for initialization, not final estimation;
(ii) it generalizes to any exponential family, not just count data;
(iii) it accounts for covariate-dependent sampling;
(iv) it produces loadings satisfying approximate estimating equations, unlike the original purely spectral null-residual approach.

\subsection{Model selection}%
\label{subsec:model_selection}

In the GMF formulation detailed in Section \ref{sec:model_specification}, the model complexity is mainly controlled by the rank of the matrix factorization, $d$.
The optimal selection of the factorization rank needs a careful balance between model complexity and goodness of fit to avoid both under- and over-fitting issues.
In the matrix factorization and factor model literature, a popular class of rank selection methods is the so-called \emph{spectral thresholding}, or \emph{elbow}, approach.
This consists of analyzing the singular values, in descending order, of a sufficiently high-rank matrix decomposition to detect a significant decrease in the rate of change of the singular values, which suggests the optimal number of factors to retain.  
Thanks to its simplicity, computational efficiency, and the significant amount of empirical and theoretical support, spectral thresholding methods gained wide popularity for rank selection problems.
See, e.g., the work of \cite{Onatski2010}, \cite{Fan2022}, \cite{Wang2023}, \cite{Liu2023}, \cite{Nie2024}, and \cite{Liu2024}.
This class of selection criteria favors a compact low-rank representation of the signal and is often used when it is of interest to identify the principal modes of variation in the data for interpretation purposes. 

Another popular rank determination method proposed in the literature leverages the concept of \emph{out-of-sample} error minimization \citep{Mazumder2010, Kidzinski2022}.
An estimate of the out-of-sample error for matrix factorization problems can be obtained either by using information-based metrics, such as the Akaike (AIC) or the Bayesian (BIC) information criteria, or prediction error measures.
The latter may be computed by removing some entries of the data matrix during the estimation process substituting them with missing values and then evaluating the reconstruction error on the hold-out set of data.
The same strategy can be performed repeatedly within a cross-validation procedure to obtain a more reliable estimate of the generalization error.
The minimizer of such an error measure over a fairly large grid of prespecified matrix ranks is the selected latent dimension.

Rank selection based on error minimization is intuitive, general, and robust with respect to the estimation method.
However, its application to high-dimensional data is hindered by its extremely high computational cost, due to the need for multiple estimations of possibly over-parametrized models.
Moreover, many state-of-the-art methods proposed in the literature and implemented in standard software packages do not handle missing values directly, thus preventing robust evaluation of out-of-sample error prior to missing value imputation.

In the high-dimensional setting, spectral thresholding methods, akin to those explored by \cite{Wang2023}, \cite{Nie2024}, and \cite{Liu2024}, gained increasing popularity, thanks to their ease of implementation, computational efficiency and connection with standard scree-plot analysis of PCA.Yet, these are not often used in practice in the omics literature, where practitioners often rely on software default values, such as 10 or 50.

Thanks to its ability of handling missing values, the intrinsic scalability of the proposed stochastic optimization method, and a convenient warm-start initialization strategy \citep{Friedman2007, Friedman2010}, our approach enables for the first time in the omics literature the systematic selection of the number of latent factors using well grounded criteria based on cross-validation out-of-sample error and spectral thresholding.

%%% SECTION
\section{Simulation studies}%
\label{sec:simulation_studies}

In this section, we assess the relative performances of the proposed estimation algorithms compared to state-of-the-art methods through several simulation experiments.
In particular, we are interested in evaluating the considered approaches in terms of execution time, memory consumption, out-of-sample prediction error, and biological signal extraction quality.

In real data scenarios, the functional form of the data-generating mechanism is typically unknown to the researcher, and the assumption of correct model specification is rarely met.
To mimic this realistic situation, in our experiments, we used different models for data simulation and signal extraction.
In this way, all the estimation methods we consider are misspecified by construction and, a priori, none of them has an advantage over the others.

%%% SUBSECTION
\subsection{Data generating process}
\label{subsec:simulations:data_generaing_process}

To simulate the data, we use the \R~package \splatter~\citep{Zappia2017}, which is freely available on \bioconductor~\citep{Huber2015}.
The package \splatter~allows us to simulate gene-expression matrices incorporating several user-specified features, such as the dimension of the matrix, the number of cell-types, the proportion of each cell-type in the sample, the form of the cell-type clusters, the expression level, the number of batches, the strength of the batch effects, and many others. 

In our experiments, we considered the following simulation setting: each dataset contains cells from five well-separated types evenly distributed in the sample. This is the signal that we aim to reconstruct with the latent factors. 
Moreover, the data are divided into three batches having different expression levels. This effect can be modeled as row covariates in our approach.
No lineage or branching effects are considered. Moreover, the simulation includes cell-specific \emph{library sizes}, which mimic the differences in the total number of counts per cell observed in real data. We model such effects with a column intercept in our framework.

To evaluate the performance of the proposed method under different regimes, we consider two simulation settings.
In simulation setting A, we compare several matrix factorization models and algorithms under a fixed latent space rank, $d = 5$, and we let the dimensions of the response matrix increase.
Specifically, we set the number of cells, $n$, to be 10 times the number of genes, $m$, and we set $m \in \{ 100, 250, 500, 750, 1000 \}$.
In simulation setting B, we compare the same set of factorization methods by fixing the dimensions of the response matrix to $n = 5000$ and $m = 500$, and letting the latent space rank grow, i.e., $d \in \{ 5, 10, 15, 20, 25 \}$.
For each combination of latent rank, $d$, number of cells, $n$, and number of genes, $m$, under the two scenarios, we generated 100 expression matrices.
Additional details are provided in Appendix C in the Supplementary Material.
Moreover, the code to generate the data and run the simulations is publicly available on GitHub\footnote{Refer to the GitHub repository \href{https://github.com/alexandresegers/sgdGMF_Paper}{\texttt{github/alexandresegers/sgdGMF\_Paper}}}.

%%% SUBSECTION
\subsection{Competing methods and performance measures}%
\label{subsec:simulations:methods_and_measures}

For the estimation, we consider several matrix factorization methods based on different model specifications coming from the statistical, machine learning, or bioinformatics literature.
In the following, we list all the methods we consider for signal extraction of the latent variables:
\begin{description}
    \setlength\itemsep{-0.025cm}
    \item {\cmf}: collective matrix factorization with non-negativity constraints and batch indicator as side information matrix (\cmfrec~package; \citealp{Cortes2023});
    \item {\nmf}: non-negative matrix factorization based on the Poisson deviance without side information matrix and automatic missing value estimation mechanism (\nmfpak~package; \citealp{Gaujoux2010});
    \item {\nmfp}: non-negative matrix factorization based on the Poisson deviance without side information matrix, where the missing values are automatically estimated together with the latent variables (\nnlm~package; \citealp{Lin2020});
    \item {\avagrad}: Poisson GMF model estimated using the AvaGrad algorithm (\glmPCA~package; \citealp{Townes2019});
    \item {\fisher}: Poisson GMF model estimated via alternated diagonal Fisher scoring (\glmPCA~package; \citealp{Townes2019});
    \item {\nbwave}: Negative Binomial GMF estimated via alternated Fisher scoring (\NewWave~pakage; \citealp{Agostinis2022});

    \item {\gfmam}: Poisson generalized factor model estimated via alternated maximization (\GFM~package; \citealp{Liu2023});

    \item {\gfmvem}: Poisson generalized factor model estimated via variational expectation maximization (\GFM~package; \citealp{Nie2024});
    
    \item {\coap}: Covariate-augmented overdispersed Poisson factor model via variational expectation maximization (\COAP~package; \citealp{Liu2024});
    \item {\airwls}: Poisson GMF model estimated via the alternated iterative re-weighted least squares algorithm of \cite{Kidzinski2022} and \cite{Wang2023} (\sgdGMF~package);
    \item {\newton}: Poisson GMF model estimated via the exact quasi-Newton algorithm of \cite{Kidzinski2022} (\sgdGMF~package);
    \item {\asgd}: Poisson GMF model estimated via the proposed adaptive stochastic gradient descent with block-wise subsampling described in Algorithm \ref{alg:block_sgd_algorithm} (\sgdGMF~package; this work)
\end{description}

Some of the most relevant characterizing features of all these methods and the relative implementations are summarized in Table \ref{tab:package_survey}.
To emulate a conventional usage of all these packages, we tried to adhere closely to the standard option setups recommended in the documentation provided by their respective authors. 
All the algorithms implemented in the \sgdGMF~package are initialized with the strategy described in Section~\ref{subsec:parameter_initialization}.
Additionally, for all methods allowing parallel computing (\cmf, \nmfp, \nbwave, \airwls~and \newton), we run the estimation employing 4 cores, a configuration commonly supported by modern PCs. For more details, we refer to Appendix~\ref{app:simulation_setting_details} in the Supplementary Material.

To account for technical confounders, we use the batch group as a covariate in all the models supporting regression effects, namely \cmf, \avagrad, \fisher, \nbwave, \coap, \airwls, \newton, and \asgd. 
Additionally, we also included row- and column-specific intercepts to account for cell- and gene-specific effects.
Thus, in our model formulation, we obtain the linear predictor $\eta_{ij} = \gamma_{0i} + \bx_{i:}^\top \bbeta_{:j} + \bu_{i:}^\top \bv_{:j}$, where $\bx_{i:}$ is a vector of dummy variables identifying which batch each cell $i$ belongs to.

To assess the relative performance of the models under consideration, we estimate them on a designated \emph{training} set and subsequently evaluate their goodness-of-fit using a \emph{validation} set.
In each simulation scenario, we initially generate a complete data matrix and then construct the training set by introducing a predetermined percentage of missing values, typically set at $30\%$. 
The positions of such holdout entries are sampled from a uniform distribution on the matrix indices.
The corresponding test set comprises all observations withheld during the training phase.
We evaluate the models in terms of elapsed execution time (in seconds), peak random access memory consumption (in megabytes), out-of-sample reconstruction, and cell-type separation in the estimated latent space.
Let $\cT$ and $\cV$ be the index sets corresponding to the training and validation entries of the response matrix, and let $\bar{y}_{_\cT}$ be the empirical average of the training matrix.
The out-of-sample reconstruction error is then computed using the relative logarithmic root mean squared error, and the relative Poisson deviance, which are defined as
\begin{equation*}
    \text{Error}(\bY, \hat\bmu) = \frac{\sum_{(i,j) \in \cV} \big[\log\{(1+y_{ij})/(1+\hat\mu_{ij})\} \big]^2}{\sum_{(i,j) \in \cV} \big[\log\{(1+y_{ij})/(1+\bar{y}_{_\cT})\} \big]^2}, \quad
    \text{Deviance} (\bY, \hat\bmu) = \frac{\sum_{(i,j) \in \cV} D(y_{ij}, \hat\mu_{ij})}{\sum_{(i,j) \in \cV} D(y_{ij}, \bar{y}_{_\cT})}.
\end{equation*}

To assess the degree of cell-type cluster separation over the estimated latent space, we consider two validation scores: the average silhouette width \citep{Rousseeuw1987} computed on a two-dimensional tSNE embedding \citep{VanderMaaten2008}, and the neighborhood purity \citep{Manning2008} evaluated on the original latent space,
which can be evaluated using the functions \texttt{silhouette()} and \texttt{neighborPurity()} from the \R~packages \texttt{cluster} \citep{Maechler2022} and \texttt{bluster} \citep{Lun2023b}, respectively.
The average silhouette is a global measure of cluster cohesion ranging from $-1$ to $1$, with $1$ indicating perfect separation between clusters.
The neighborhood purity is a local measure of cluster cohesion ranging from $0$ to $1$, with $1$ indicating perfect coherence between clusters.
Being a global measure based on Euclidean distances, the silhouette favors clusters with spherical shapes, while it is not able to detect well-separated clusters featuring non-spherical or non-convex boundaries.
On the other hand, neighborhood purity is a local measure of cluster separation, which is able to detect localized behavior and does not depend on the cluster shapes.

%%% SUBSECTION
\subsection{Simulation results}%
\label{subsec:simulations:simulation_results}

\begin{figure}
    \centering
    \includegraphics[width = 
    \textwidth]{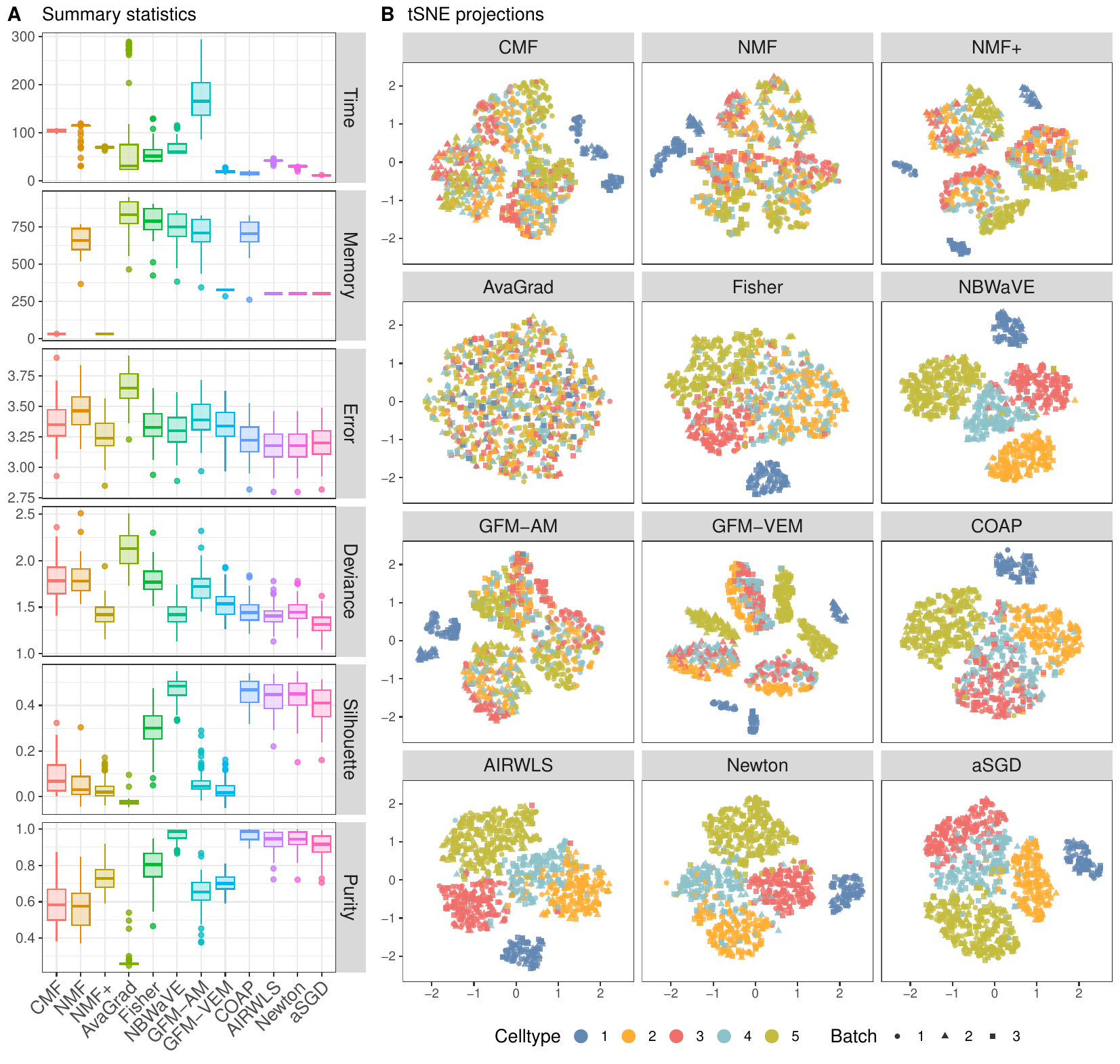}
    \caption{Summary information for the simulation experiment described in Section \ref{subsec:simulations:data_generaing_process} with $n = 5000$, $m = 500$ and $d = 5$. 
    Left (panel A): summary statistics reporting the execution time (in seconds), the peak memory consumption (in megabytes), the out-of-sample relative logarithmic root mean squared error (multiplied by $100\%$), the out-of-sample relative Poisson deviance (multiplied by $100\%$), the silhouette score of the true cell-type clusters calculated on a 2-dimensional tSNE projection, and the mean cluster purity of the true cell-type clusters calculated on the 5-dimensional estimated latent space.
    Right (panel B): 2-dimensional tSNE projections of the estimated latent factors for one specific replication of the experiment.
    }
    \label{fig:static_simulation_summary}
\end{figure}

Figure~\ref{fig:static_simulation_summary} presents an overview of the results obtained from the data simulated with $n=5000$, $m=500$, and $d=5$. 
Our proposed \asgd~method is the fastest approach together with \gfmvem~and \coap~(Fig.~\ref{fig:static_simulation_summary}A): it is parsimonious in terms of memory usage, while performing similarly to the best-performing methods in terms of out-of-sample error, deviance, silhouette width, and neighborhood purity. 
An exemplar simulation run (Fig.~\ref{fig:static_simulation_summary}B) shows that all the GMF methods using covariate information succeed in reconstructing the original cell-type clustering, being able to filter out the batch effect via the regression term in the linear predictor.
The only exception is \avagrad, which fails to separate the different groups. 
On the other hand, \gfmam~and \gfmvem~are unable to disentangle the cell-types from the batch effects, primarily due to the \GFM~package's inability to account for covariate effects.
Similarly, the NMF methods over-cluster the data, not allowing for batch effect removal via regression.
This is confirmed by the average silhouette width and the mean neighborhood purity (bottom lines of Fig.~\ref{fig:static_simulation_summary}A), which show a large difference in performance, discriminating methods based on their ability to account for covariate (i.e., batch) effects.

These results are confirmed across both setting A and B (Fig. \ref{fig:dynamic_simulation_time} and \ref{fig:dynamic_simulation_summary}): in terms of computational efficiency, the proposed \asgd~implementation consistently outperforms its competitors in setting A and ranks as the top performer in setting B, showing lower elapsed execution times, and a better scalability with respect to both the sample size and the dimension of the latent space (top row of Fig.~\ref{fig:dynamic_simulation_time}). 
Additionally, \asgd~manifests a parsimonious management of the random access memory, which is aligned to \airwls, \newton, and \gfmvem~and surpassed only by \cmf~and \nmfp~(bottom row of Fig.~\ref{fig:dynamic_simulation_time}).
Interestingly, \cmf, \nmfp, \gfmvem, \airwls, \newton, and \asgd~maintain an almost constant memory footprint in simulation setting B. 
This behavior is attributed to their efficient memory allocation strategies, where memory usage is mainly dictated by the storage of input data and sufficient statistics, with negligible additional cost during optimization, even as the latent dimension scales but is much smaller than the matrix dimensions.
Among the alternative methods, \coap~and \newton~emerge as the fastest, followed by \airwls.
On the other hand, \gfmam, \nbwave, \fisher~and \avagrad~optimizers display inferior scalability compared to \asgd, \newton, and \airwls~(Fig. \ref{fig:dynamic_simulation_time}).
Contrary to the claims in the \glmPCA~package documentation, our findings indicate that the \avagrad~optimizer often lacks in both speed and reliability when compared to the \fisher~optimizer, which typically achieves convergence within a reasonable time and, on average, runs faster than \avagrad.

Regarding non-negative matrix factorization, implementations such as \nmf~and \cmf~ demonstrate poor computational scalability across both settings. 
However, \nmfp~shows improved performance, reaching efficiency levels close to \airwls~in setting B, and an optimal management of the memory in both settings (Fig. \ref{fig:dynamic_simulation_time}).

In terms of goodness-of-fit measures (top two rows of Fig.~\ref{fig:dynamic_simulation_summary}), the least out-of-sample logarithmic error and deviance are systematically achieved by \nmfp, \coap, \airwls, \newton, and \asgd, with \asgd~consistently emerging as the most accurate method in the deviance metric.
Such a situation partially reflects on the silhouette width and neighborhood purity scores (bottom two rows of Fig. \ref{fig:dynamic_simulation_summary}), for which, in both scenarios A and B, \nbwave, \coap, \airwls, \newton, and \asgd~always outperform the other methods in terms of cell-type separation in the latent space.
As such, \asgd~exhibits slightly suboptimal performance compared to the other top-performing methods, especially in small sample settings; however, it converges towards them as the sample size increases.
This behavior is not unexpected, as stochastic optimization algorithms typically require a large number of observations to mitigate their intrinsic randomness, achieving stability as the data dimension grows.

\begin{figure}
    \centering
    \includegraphics[width = \textwidth]{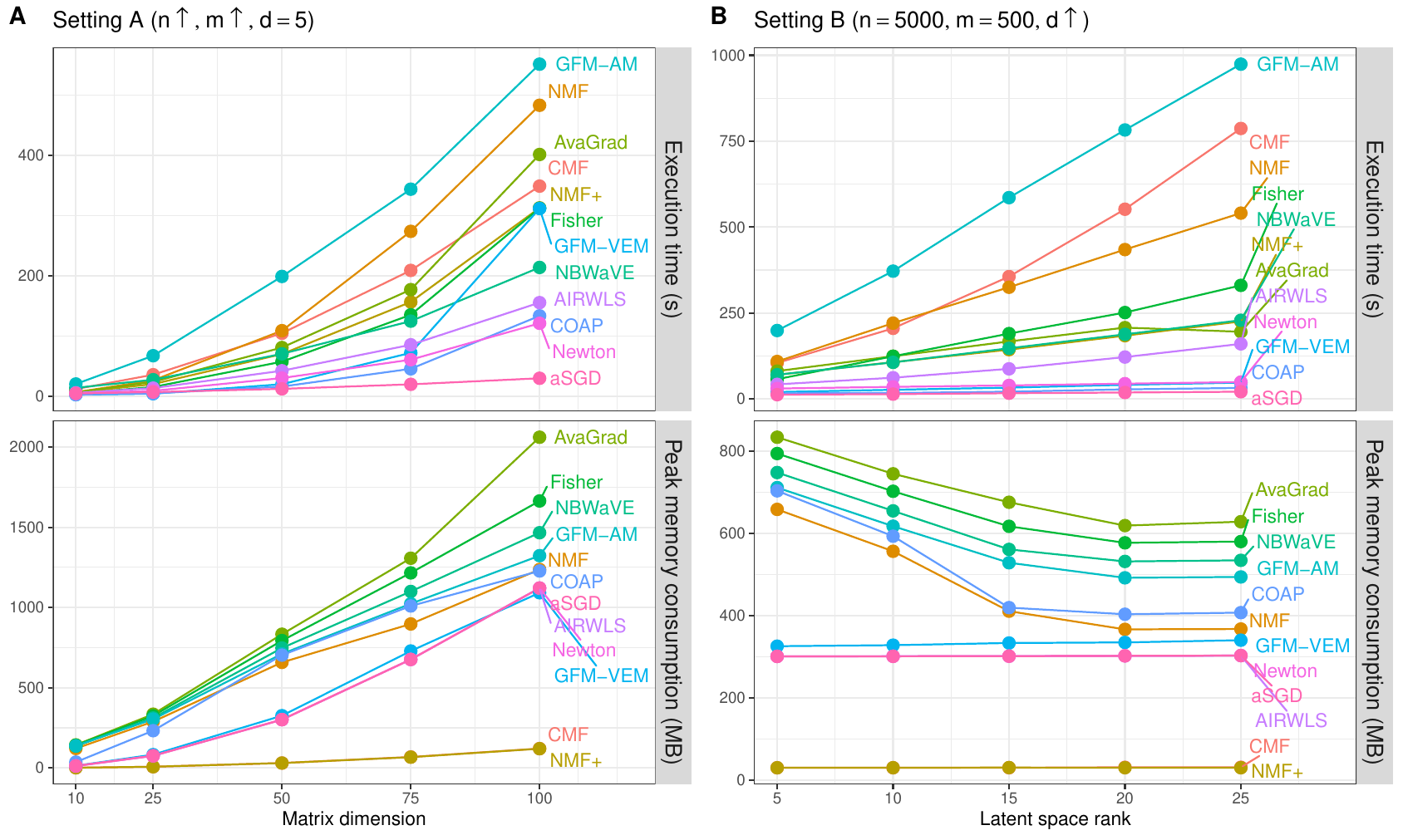}
    \caption{%
    Summary statistics of the simulation experiments described in Section \ref{subsec:simulations:data_generaing_process}.
    The columns correspond to simulation settings A (left) and B (right). 
    The rows correspond to the elapsed execution time in seconds (top) and the peak memory consumption in megabytes (bottom). 
    }
    \label{fig:dynamic_simulation_time}
\end{figure}

\begin{figure}
    \centering
    \includegraphics[width = \textwidth]{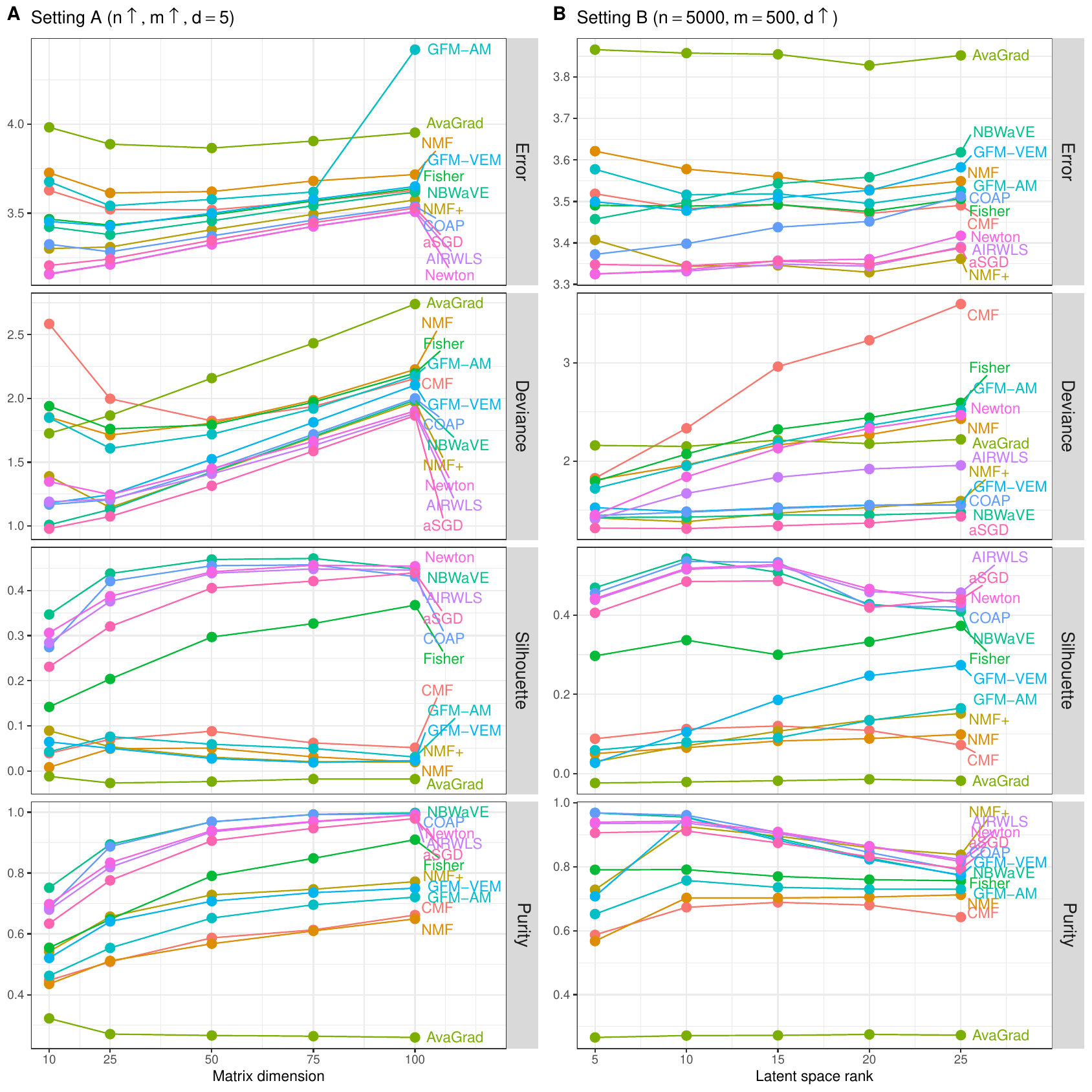}
    \caption{%
    Summary statistics of the simulation experiments described in Section \ref{subsec:simulations:data_generaing_process}.
    The columns correspond to simulation settings A (left) and B (right).
    The rows correspond to four goodness-of-fit measures. 
    From top to bottom: the out-of-sample relative logarithmic root mean squared error, the out-of-sample relative residual deviance, the silhouette evaluated on a 2-dimensional tSNE projection of the latent space, and the neighborhood purity of the true cell-type evaluated on the original latent space.
    }
    \label{fig:dynamic_simulation_summary}
\end{figure}

%%% SECTION
\section{Real data applications}%
\label{sec:real_data_applications}

In this section, we will demonstrate the effectiveness of our method using two real datasets. The first dataset, referred to as the Arigoni dataset, is a 10X Genomics scRNA-seq experiment on lung cancer cell lines with unique driver mutations \citep{Arigoni2024}. As suggested by the authors, the heterogeneity among cell lines can be used as ground truth to benchmark computational methods. We use these data to illustrate that our method can discover real biological signal. 

Further, we apply our method to a large scRNA-seq dataset consisting of more than 1.3 million cells from the mouse brain, generated by 10X Genomics \citep{Lun2023}. As this dataset does not contain a ground truth of cell labels, this dataset primarily showcases the scalability of our approach in large datasets.

Throughout this section, we consider parametrization (B1) to obtain an orthonormal loading matrix $\bV$ and a scaled orthogonal score factor matrix $\bU$ (see Section \ref{subsec:parameter_identifiability}). This choice is conventional in the RNA-seq literature \citep[see, e.g.,][]{Risso2018, Townes2019, Ahlmann2023} and is coherent with the standard parametrization of principal component analysis.

%%% SUBSECTION
\subsection{Arigoni data}
\label{subsec:applications:arigoni_data}

The original Arigoni dataset \citep{Arigoni2024} consists of $29{,}606$ cells from 8 different lung cancer cell lines with unique driver mutations (EGFR, ALK, MET, ERBB2, KRAS, BRAF, ROS1). The ground truth knowledge of the driver mutations can be used to evaluate cell clustering by visual inspection and upon using clustering algorithms. 
Importantly, the CCL-185-IG cell line is derived from the A549 cell line. 
Therefore, only subtle differences are expected between these two cell lines. 
Quality control is done using the \texttt{perCellQCFilters} function of the \R~package \texttt{scuttle} \citep{McCarthy2017}, which filtered cells that have a library size lower than 1306, a percentage of mitochondrial reads higher than 6.05\% or cells that have fewer than 732 features expressed. 
Additionally, peripheral blood mononuclear cells are removed due to their distinct expression profile. The final filtered dataset includes $26{,}472$ cells from 7 different cell lines. 
Unless mentioned differently, all the analyses are based on the 500 most variable genes, a common choice in standard scRNA-seq workflows. The selection of the number of highly-variable genes did not prove critical, as the results are robust across a large spectrum of values (Supplementary Fig. \ref{fig:hvg} and \ref{fig:hvg_purity}).

To select the optimal matrix rank for the latent component of the model, different model selection criteria are assessed in a 5-fold cross validation (see Section \ref{subsec:model_selection} for details). 
AIC and BIC are assessed on the 5 training data partitions, while the out-of-sample deviance is calculated on the test data. 
Also, considering all the data, we assess the scree plot of the eigenvalues based on the deviance residuals after using OLS on the log-transformed data.

Both the AIC and out-of-sample deviance criteria suggest a matrix rank of 15, while the BIC and the scree plot suggest 9 (Fig. \ref{fig:model_selection_arigoni}A). 
The mean cell-line neighborhood purity scores (Fig. \ref{fig:model_selection_arigoni}B) show that, for the majority of cell lines, a rank of 9 is sufficient to completely separate the groups, with only cell lines A549 and CCL-185-IG exhibiting a score lower than $0.9$. 
At rank 15, all cell lines achieve a large value of mean purity, and increasing the rank further does not improve this index, while introducing the risk of overfitting the data, as observed in the increasing out-of-sample deviance in the cross-validation (Fig. \ref{fig:model_selection_arigoni}A). 
These remarks are confirmed when visually inspecting the tSNE plots, coloured by the ground truth labels (Fig. \ref{fig:model_selection_arigoni}C), that show that a matrix rank of 9 yields a good separation, except for A549 from CCL-185-IG, which show a slight degree of mixing (Fig. \ref{fig:model_selection_arigoni}C, confusion matrix). 
Although there are no outstanding visual differences between rank 15 and 30, performing Leiden clustering with a resolution tuned to obtain 7 clusters shows that a matrix of rank 30 results in one small cluster with very few cells, rather than separating the A549 and CCL-185-IG groups, while the latter happens as expected for rank 15. 
This suggests that 15 is a reasonable number of latent factors to include in the model.
Importantly, the tested model selection approaches were not informed by cell type labels, which were used only to evaluate the methods' performance. This shows that unsupervised approaches to model selection are able to estimate a number of factors sufficient to extract meaningful biological signal from real data.

\begin{figure}
    \centering
    \includegraphics[width = 
    \textwidth]{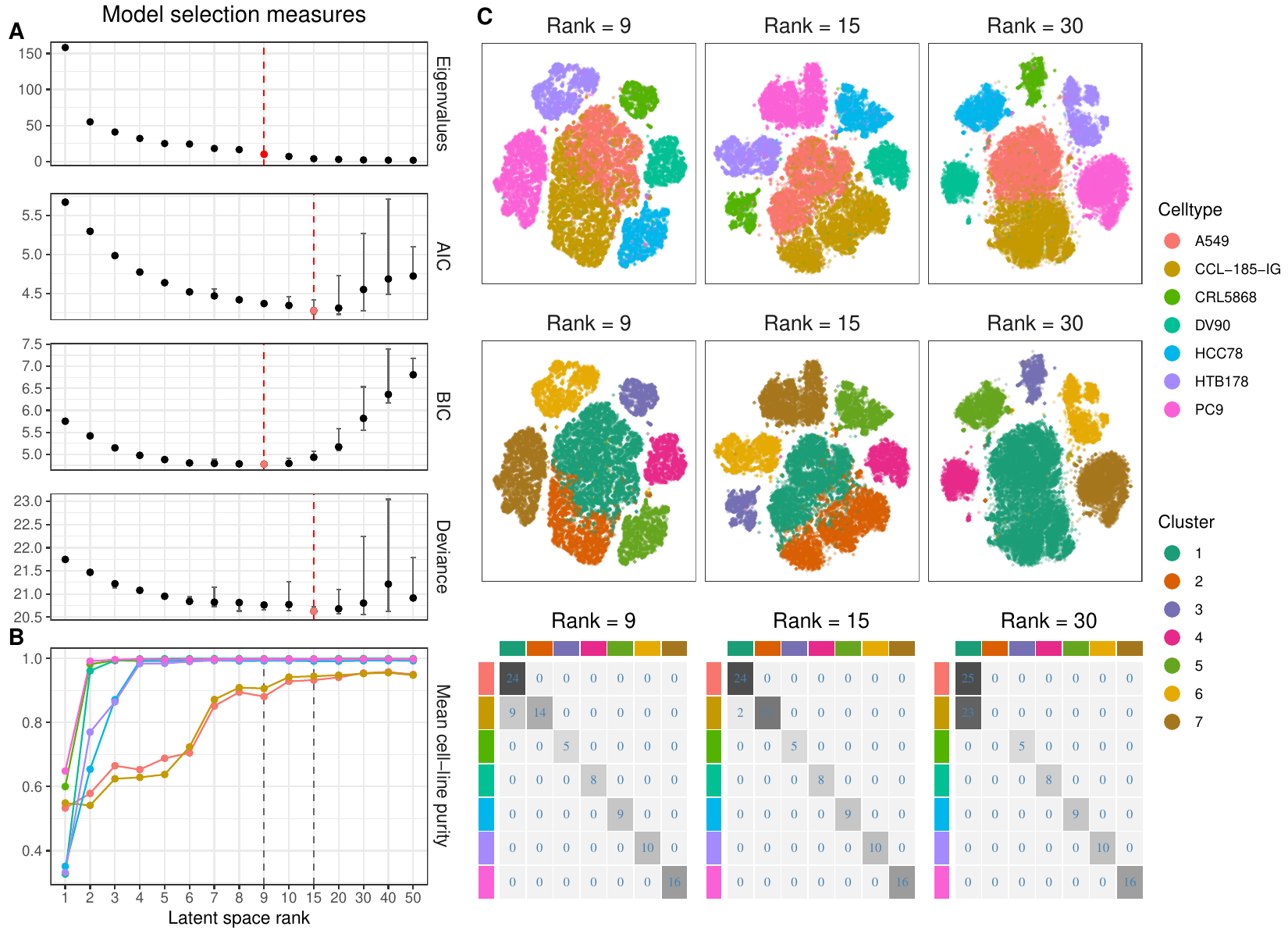}
    \caption{%
    	\label{fig:model_selection_arigoni}%
    	Assessment of model selection metrics. 
    	A) Application of diverse model selection criteria including scree plot, AIC, BIC, and cross-validation based on out-of-sample deviances. 
    	B) Mean cell-line purity as a function of the matrix rank. 
    	C) tSNE plots coloured by the ground truth and clusters obtained by Leiden clustering, alongside a confusion matrix representing cell-line distributions across clusters. 
    	In the confusion matrices, each entry is featured with a number corresponding to the percentage of cells belonging to that configuration. 
    	The matrix total is equal to 100 and the colour intensity is proportional to the percentages. 
    }
\end{figure}

Using 15 factors, as suggested by the model selection criteria, \asgd~was compared with \nbwave, \fisher, and \avagrad, three of the most popular methods in single-cell analysis, as well as with \coap, one of the fastest and most reliable method in the simulations in addition to \asgd~(Supplementary Fig. \ref{fig:Arigoni-newwave-comparison}). 
We tested \asgd~with both a Negative Binomial and a Poisson likelihood, and the results were virtually identical (Supplementary Fig. \ref{fig:Arigoni-newwave-comparison}). 
All methods achieve similar results, as shown by both visual inspection of the tSNE plots and the mean cell-line neighborhood purities. 
However, \asgd~is orders of magnitude faster, and therefore allows for model selection obtaining optimal matrix rank, which was shown to be important for clustering the data. 
Further, \asgd~has a better out-of-sample deviance. 
In terms of memory usage, \asgd~has a similar peak RAM memory usage as \coap, outperforming the three other methods.

\subsection{TENxBrainData}
\label{subsec:applications:tenxbraine_data}

To demonstrate the scalability of our method to large datasets, we apply \asgd~to the TENxBrainData \citep{Lun2023}, which consists of scRNA-seq UMI counts generated by 10X Genomics for approximately 1.3 million cells. 
Cells of this dataset were obtained from the cortex, hippocampus and ventricular zone of two mouse brains. 
Although no ground truth is available for the different cell types in this dataset, marker genes that can discriminate between subtypes of cells are available for mouse brains, and the list used by \cite{Hicks2021} is used to qualitatively evaluate extraction of biological signal in different cell clusters. 
Quality control and filtering are performed using the \R~package \texttt{scater} \citep{McCarthy2017}, excluding cells with an exceptional number of mitochondrial reads (more than 3 median absolute deviances away from the median), and genes with no expression in over 99\% of the cells. This procedure yields $1{,}232{,}055$ cells and we retain the 500 most variable genes, as done with the Arigoni dataset.

Considering the eigenvalue gap method (Supplementary Fig. \ref{fig:casestudy-screeplot}), which is a very fast procedure for model selection, we selected a model with 10 latent factors. 
To study the scalability of \asgd~and its competitors on this large dataset, we considered subsamples of $100{,}000$, $200{,}000$ and $300{,}000$ cells (Fig. \ref{fig:casestudy}A). This analysis showed that \asgd~is orders of magnitude faster than competing methods, with \fisher~and \nbwave~taking $4$ and $8$ hours, respectively, to analyze $300{,}000$ cells. 
While \avagrad~achieves better computational speed than \fisher~and \nbwave, it remains extremely slow compared to \asgd. Moreover, \avagrad~did not converge one in five times for $100{,}000$ cells, and two in five times for $200{,}000$ cells. 
Therefore, these methods would result in extreme computational times on the full dataset. Moreover, the memory usage of \asgd~remains lower than competing methods. 
Note also that \coap~returned errors when using more than $70{,}000$ cells, rendering it unable to handle large datasets. Therefore, only \asgd~is a reasonable model to run on the full dataset due to its superior computational efficiency.

Using \asgd~on the full dataset returned results in 77 minutes. 
Subsequent Leiden clustering of this matrix factorization revealed sensible biological signal extraction, as its clusters align with established marker genes (Fig. \ref{fig:casestudy}B). 
For example, cluster 8 is characterized by cells expressing Pyramidal neuron cells marker genes, such as Crym \citep{Loo2019}, while cluster 4 contains cells expressing interneuron markers, e.g., Sst and Lhx6  \citep{Tasic2018}. 

\begin{figure}
    \centering
    \includegraphics[width = 
    \textwidth]{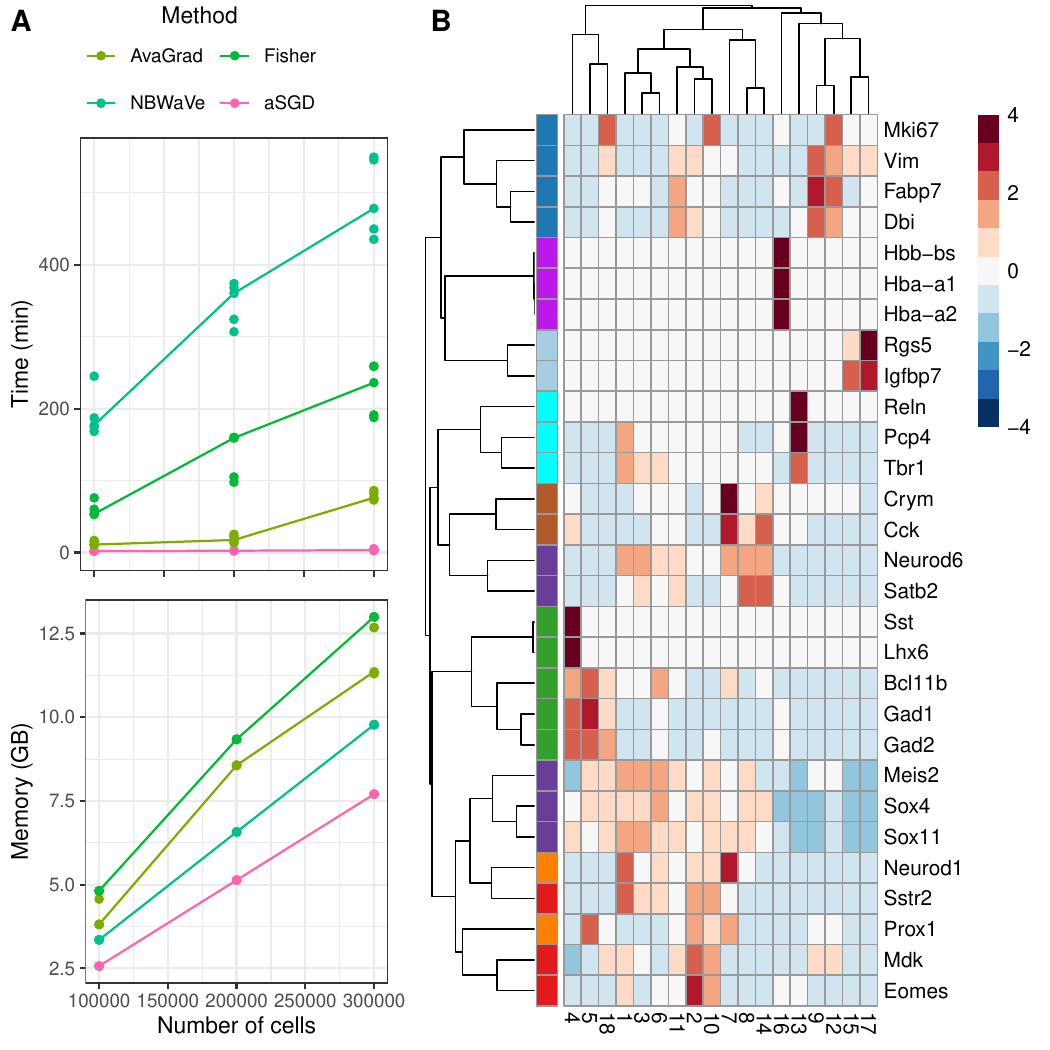}
    \caption{%
    	\label{fig:casestudy}%
    	Results on large-scale data. 
    	A) Computational time and memory usage for different methods on increasingly larger subsets of the dataset.
    	Each subset includes 500 high-variable genes and a growing number of cells. 
    	B) Heatmap of the average gene expression of the 18 clusters obtained by Leiden clustering computed on the latent score matrix, for 29 marker genes. Each marker gene is coloured based on the cell type it is expressed in.
    }
\end{figure}

% \sout{To study the scalability of \asgd~and its competitors on this real dataset, we considered subsamples of $100{,}000$, $200{,}000$ and $300{,}000$ cells (Fig. \ref{fig:casestudy}B).  This analysis showed that \asgd~is orders of magnitude faster than competing methods, with \fisher~and \nbwave~taking $3.5$ and $10$ hours, respectively, to analyze $300{,}000$ cells.  While \avagrad~achieves better computational speed than \fisher~and \nbwave, it still takes more for it to analyze $300{,}000$ cells than for \asgd~to fit the whole 1.2 million cells. Moreover, \avagrad~did not converge one in five times for $100{,}000$ cells, and two in five times for $200{,}000$ cells.}

%%% SECTION
\section{Discussion}%
\label{sec:discussion}

In the present work, we propose a flexible and scalable tool to perform generalized matrix factorization in massive data problems, with a particular focus on scRNA-seq applications. 
We propose an innovative adaptive stochastic gradient descent algorithm, whose performances are enhanced via a memory-efficient block-wise subsampling method and a convenient initialization strategy. An \R/\Cpp~implementation in the open-source package \texttt{sgdGMF} is freely available on CRAN. 
Overall, the proposed method proved competitive with state-of-the-art approaches available in \R, showing higher prediction accuracy, good biological signal-extraction ability, and a significant speed-up of the execution time in simulated and real data examples.
Unlike most methods currently employed for scRNA-seq signal extraction, our approach natively deals with missing values, iteratively imputing them with the model's current best prediction. 
This feature proved important for out-of-sample error evaluation, model selection, and matrix completion.

An appealing feature of the proposed method is its flexibility, which enables several extensions and generalizations. 
The proposed framework naturally extends to heterogeneous likelihood specification across rows or columns of the response matrix. 
This would make it possible to jointly factorize discrete, count and continuous data sharing the same latent factorization structure, but different conditional distributions. 
From a biological viewpoint, this extension would permit to flexibly model multi-omic data \citep{Argelaguet2018, Argelaguet2020} under the unified framework provided by \texttt{sgdGMF}.

Another interesting extension of the proposed algorithm is generalized tensor factorization for non-Gaussian data arrays. In this setting, the computational complexity of estimating highly parametrized models can easily grow very fast. Thus, cheap and modulable estimation algorithms, such as the proposed adaptive stochastic gradient descent, are becoming increasingly important.

From an algorithmic viewpoint, another fascinating possibility is considering non-uniform sampling schemes for the mini-batch selection. For instance, if the sample is divided into known subpopulations, it could be convenient to exploit the clustered nature of the data when forming the mini-batch partition via stratification. This can improve the representativeness of each chunk, reduce the variance of the gradient estimator, and prevent the optimization from converging to suboptimal maxima dominated by a specific subpopulation signal.

\section{Funding}

This work was supported by EU funding within the MUR PNRR ``National Center for HPC, big data and quantum computing'' (Project no. CN00000013 CN1). 
The views and opinions expressed are only those of the authors and do not necessarily reflect those of the European Union or the European Commission. 
Neither the European Union nor the European Commission can be held responsible for them. 
DR was also supported by the National Cancer Institute of the National Institutes of Health (U24CA289073) and by project EOSS6-0000000644 from the Chan Zuckerberg Initiative. 
This work was supported by grants from Ghent University Special Research Fund (BOF20/GOA/023) (A.S., L.C.), Research Foundation Flanders (FWO G062219N) (A.S., L.C.).

\section{Data availability}

\texttt{sgdGMF} is freely available as an open-source \R~package at \davide{\href{https://CRAN.R-project.org/package=sgdGMF}{\texttt{CRAN.R-project.org/package=sgdGMF}}}.
The scripts used to run all analyses are available on GitHub at \href{https://github.com/alexandresegers/sgdGMF_Paper}{\texttt{github/alexandresegers/sgdGMF\_Paper}}.
The Arigoni dataset is available at \href{https://doi.
org/10.6084/m9.figshare.23939481.v1}{\texttt{https://doi.org/10.6084/m9.figshare.23939481.v1}}. 
The TENxBrainData dataset is available as part of the \texttt{TENxBrainData} Bioconductor package at \href{https://bioconductor.org/packages/TENxBrainData}{\texttt{https://bioconductor.org/packages/TENxBrainData}}.

%%% REFERENCES
\bibliographystyle{apalike}
\bibliography{./biblio}

%%% APPENDIX
\clearpage

\begin{center}\Large\bf
    Appendices of ``Stochastic gradient descent estimation of \\ 
    generalized matrix factorization models with \\
    application to single-cell RNA sequencing data''
\end{center}

\begin{appendix}

%%% APPENDIX A
\section{Parameter identifiability}
\label{app:parameter_identifiability}

Here, we discuss the set of identifiability constraints introduced in Section~\ref{subsec:parameter_identifiability}, proving the model identification under constraints (A) and (B), and providing a practical and efficient post-processing approach to enforce the identification of the parameter estimates.

\begin{proposition}
	\label{prop:model_identifiability}
	Under assumption (A) and any of the equivalent conditions (B1), (B2), or (B3), the GMF model in \eqref{eq:exponential_family_distribution}--\eqref{eq:linear_predictor_model} is computationally identifiable.
\end{proposition}

To streamline the proof of Proposition \ref{prop:model_identifiability}, we first state and demonstrate the following Lemmas.

\begin{lemma}
	\label{lmm:identifiability_under_A}
	Let $(\bB_1, \bGamma_1, \bU_1, \bV_1)$ and $(\bB_2, \bGamma_2, \bU_2, \bV_2)$ be two configurations of the parameters, both satisfying (A) and yielding the same linear predictor.
	Then $\bB_1 = \bB_2$, $\bGamma_1 = \bGamma_2$, and $\bU_1 \bV_1^\top = \bU_2 \bV_2^\top$.
\end{lemma}

\begin{proof}
	By hypothesis, $(\bB_1, \bGamma_1, \bU_1, \bV_1)$ and $(\bB_2, \bGamma_2, \bU_2, \bV_2)$ yield the same linear predictor, that is
	\begin{equation}
		\label{eq:paired_linear_predictors}
		\bet_1 = \bX \bB_1^\top + \bGamma_1 \bZ^\top + \bU_1 \bV_1^\top 
		\;=\; 
		\bX \bB_2^\top + \bGamma_2 \bZ^\top + \bU_2 \bV_2^\top = \bet_2.
	\end{equation}
	Left-multiplying both sides of \eqref{eq:paired_linear_predictors} by $\bX^\top$, and using of the constraints $\bX^\top \bGamma_1 = \bX^\top \bGamma_2 = \bzero$ and $\bX^\top \bU_1 = \bX^\top \bU_2 = \bzero$, we obtain $(\bX^\top \bX) \bB_1 = (\bX^\top \bX) \bB_2$. Thanks to the non-singularity of $\bX^\top \bX$, we conclude that $\bB_1 = \bB_2$. 
	
	Subtracting the common term $\bX \bB_1^\top$ (equal to $\bX \bB_2^\top$) from both sides of \eqref{eq:paired_linear_predictors}, we obtain $\bGamma_1 \bZ^\top + \bU_1 \bV_1^\top = \bGamma_2 \bZ^\top + \bU_2 \bV_2^\top$.
	Right-multiplying by $\bZ$, and using the constraints $\bZ^\top \bV_1 = \bZ^\top \bV_2 = \bzero$, we get $\bGamma_1 (\bZ^\top \bZ) = \bGamma_2 (\bZ^\top \bZ)$.
	This, along with the non-singularity of $\bZ^\top \bZ$, implies $\bGamma_1 = \bGamma_2$.
	
	Moreover, given the identities $\bB_1 = \bB_2$ and $\bGamma_1 = \bGamma_2$, equation \eqref{eq:paired_linear_predictors} entails $\bU_1 \bV_1^\top = \bU_2 \bV_2^\top$.
	This concludes the proof.
\end{proof}

\begin{lemma}
	\label{lmm:identifiability_under_B1_B2}
	Let $(\bU_1, \bV_1)$ and $(\bU_2, \bV_2)$ be two configurations of the score and loading parameters, both satisfying either (B1) or (B2), and $\bU_1 \bV_1^\top = \bU_2 \bV_2^\top$.
	Then $\bU_1 = \bU_2$ and $\bV_1 = \bV_2$.
\end{lemma}

\begin{proof}
	First, we consider assumption (B1).
	Right-multiplying $\bU_1 \bV_1^\top = \bU_2 \bV_2^\top$ on both sides by $\bV_1$, and exploiting $\bV_1^\top \bV_1 = \bI_d$, we obtain $\bU_1 = \bU_2 (\bV_2^\top \bV_1)$.
	Let $\bA = \bV_2^\top \bV_1$, then, substituting $\bU_1 = \bU_2 \bA$ into $\bU_1 \bV_1^\top = \bU_2 \bV_2^\top$, left-multiplying the result by $\bU_2^\top$, and using the non-singularity of $\bU_2^\top \bU_2$, we obtain $\bV_2 = \bV_1 \bA^\top$.
	
	Sequentially applying $\bV_2^\top \bV_2 = \bI_d$, $\bV_2 = \bV_1 \bA^\top$, and $\bV_1^\top \bV_1 = \bI_d$, we have
	\begin{equation*}
		\bI_d = \bV_2^\top \bV_2 = \bA \bV_1^\top \bV_1 \bA^\top = \bA \bI_d \bA^\top = \bA \bA^\top,
	\end{equation*}
	which implies that $\bA$ is an orthogonal matrix.
	
	Similarly, we sequentially use $\bU_1^\top \bU_1 = \bSigma$, $\bU_1 = \bU_2 \bA$, and $\bU_2^\top \bU_2 = \bSigma$ to get
	\begin{equation*}
		\bSigma = \bU_1^\top \bU_1 = \bA^\top \bU_2^\top \bU_2 \bA = \bA^\top \bSigma \bA.
	\end{equation*}
	Thus, $\bA^\top \bSigma \bA$ is the eigenvalue decomposition of $\bSigma$, where $\bA$, due to the orthogonality, must be a diagonal sign matrix whose diagonal entries can only take values $-1$ or $+1$. 
	Finally, since $\bV_2 = \bV_1 \bA^\top$ and the first non-zero entry of each column of $\bV_1$ and $\bV_2$ must be positive because of (B1), we have $\bA = \bI_d$.
	As a consequence, $\bU_1 = \bU_2$ and $\bV_1 = \bV_2$. This concludes the proof for (B1).
	
	The proof under (B2) trivially follows from the same argument.
\end{proof}

\begin{lemma}
	\label{lmm:identifiability_under_B3}
	Let $(\bU_1, \bV_1)$ and $(\bU_2, \bV_2)$ be two configurations of the score and loading parameters, both satisfying (B3) and $\bU_1 \bV_1^\top = \bU_2 \bV_2^\top$.
	Then $\bU_1 = \bU_2$ and $\bV_1 = \bV_2$.
\end{lemma}

\begin{proof}
	Following the same reasoning used in the proof of Lemma \ref{lmm:identifiability_under_B1_B2}, and defining $\bH_n = \frac{1}{n} (\bI_n - \frac{1}{n} \bone_n \bone_n^\top)$, we can left-multiply $\bU_1 \bV_1^\top = \bU_2 \bV_2^\top$ on both sides by $\bU_1^\top \bH_n$ and use $\bU_1^\top \bH_n \bU_1 = \bI_d$ to obtain the identity $\bV_1 = \bV_2 \bA^\top$, where $\bA = \bU_1^\top \bH_n \bU_2$.
	Substituting $\bV_1 = \bV_2 \bA^\top$ into $\bU_1 \bV_1^\top = \bU_2 \bV_2^\top$, right-multiplying by $\bV_2$, and exploiting the non-singularity of $\bV_2^\top \bV_2$, we obtain $\bU_2 = \bU_1 \bA$.
	Hence, sequentially applying $\bU_2^\top \bH_n \bU_2 = \bI_d$, $\bU_2 = \bU_1 \bA$, and $\bU_1^\top \bH_n \bU_1 = \bI_d$, we have
	\begin{equation*}
		\bI_d 
		= \bU_2^\top \bH_n \bU_2
		= \bA^\top \bU_1^\top \bH_n \bU_1 \bA
		= \bA^\top \bI_d \bA = \bA^\top \bA,
	\end{equation*}
	which entails that $\bA$ is orthogonal.
	
	Now, we know that $\bV_1 = \bV_2 \bA^\top$, where $\bV_1$ and $\bV_2$ are $m \times d$ matrices with zero upper-triangular part, and positive diagonal. 
	Then, defining $\bR_1$ and $\bR_2$ as the $d \times d$ submatrices containing the first $d$ rows of $\bV_1$ and $\bV_2$, respectively, we have $\bR_1 = \bR_2 \bA^\top$.
	By construction, $\bR_1$ and $\bR_2$ are non-singular squared lower-triangular matrices; thus we can write $\bR_2^{-1} \bR_1 = \bA^\top$.
	As $\bR_2$ and $\bR_1$ are lower-triangular with positive diagonal entries, $\bR_2^{-1} \bR_1$ also inherits the same properties, as well as $\bA^\top$, which is also orthogonal.
	Recall that orthogonal triangular matrices are, in fact, diagonal sign matrices, then $\bA$ is diagonal with $-1$ or $+1$ diagonal entries.
	But $\bA^\top = \bR_2^{-1} \bR_1$ must also have positive diagonal entries, hence $\bA = \bI_d$. 
	And, finally, we can establish that $\bU_1 = \bU_2$ and $\bV_1 = \bV_2$. This concludes the proof.
\end{proof}

Now, we are ready to prove Proposition \ref{prop:model_identifiability}.

\begin{proof}[Proof of Proposition \ref{prop:model_identifiability}]
	The likelihood function of model~\eqref{eq:exponential_family_distribution}--\eqref{eq:linear_predictor_model} depends on the parameters $(\bB, \bGamma, \bU, \bV)$ solely through the linear predictor~\eqref{eq:linear_predictor_model}.
	Then, if we can prove that, under (A) and (B), equal linear predictors imply equal parameter values, this automatically ensures identifiability.
	
	Suppose that $(\bB_1, \bGamma_1, \bU_1, \bV_1)$ and $(\bB_2, \bGamma_2, \bU_2, \bV_2)$ are two parameter configurations, both satisfying assumption (A) and yielding the same linear predictor.
	Then, thanks to assumption (A) and Lemma \ref{lmm:identifiability_under_A}, we have $\bB_1 = \bB_2$, $\bGamma_1 = \bGamma_2$, and $\bU_1 \bV_1^\top = \bU_2 \bV_2^\top$.
	This means that assumption (A) alone is enough to ensure the identifiability of $\bB$ and $\bGamma$. 
	Now, we are left to prove the identifiability of $\bU$ and $\bV$ under one of the equivalent additional restrictions  (B1), (B2), and (B3).
	
	Assuming that $(\bU_1, \bV_1)$ and $(\bU_2, \bV_2)$ satisfy both (A) and (B1), from Lemma \ref{lmm:identifiability_under_B1_B2} we have $\bU_1 = \bU_2$ and $\bV_1 = \bV_2$.
	The same holds under (A) and (B2).
	This ensures the identifiability of $(\bB, \bGamma, \bU, \bV)$ under assumptions (A) and (B1), or (A) and (B2).
	
	Finally, assuming that $(\bU_1, \bV_1)$ and $(\bU_2, \bV_2)$ satisfy both (A) and (B3), from Lemma \ref{lmm:identifiability_under_B3} we have $\bU_1 = \bU_2$ and $\bV_1 = \bV_2$.
	This ensures the identifiability of $(\bB, \bGamma, \bU, \bV)$ under assumptions (A) and (B3).
	This concludes the proof.
\end{proof}

To obtain identifiable estimates that satisfy (A) and any of (B1), (B2), or (B3), we need a stable and efficient post-processing method to project unrestricted solutions onto the constrained space induced by the identifiability restrictions.
To this end, we can apply standard projection methods. 
Consider, for instance, the constraint $\bX^\top \bGamma = \bzero$, and define $\bP_{\scx} = \bX (\bX^\top \bX)^{-1} \bX^\top$ as the projection matrix onto the column space of $\bX$.
As it is evident, $\bGamma = \bP_{\scx} \bGamma + (\bI_n - \bP_{\scx}) \bGamma$, hence
\begin{align*}
	\bX \bB^\top + \bGamma \bZ^\top 
	& = \bX \bB^\top + \bP_{\scx} \bGamma \bZ^\top + (\bI_n - \bP_{\scx}) \bGamma \bZ^\top \\
	& = \bX \bB^\top + \bX \big[ (\bX^\top \bX)^{-1} \bX^\top \bGamma \bZ^\top \big] + \big[ \,\bI_n - \bX (\bX^\top \bX)^{-1} \bX^\top \big] \bGamma \bZ^\top \\
	& = \bX \big[ \bB + \bZ \bGamma^\top \bX (\bX^\top \bX)^{-1} \big]{}^\top + \big[ \bGamma - \bX (\bX^\top \bX)^{-1} \bX^\top \bGamma \big] \bZ^\top.
\end{align*}
Therefore, we can reparametrize both $\bB$ and $\bGamma$ as follows
\begin{equation*}
	\bB^* = \bB + \bZ \big[ (\bX^\top \bX)^{-1} \bX^\top \bGamma \big]{}^\top, \quad
	\bGamma^* = \bGamma - \bX \big[ (\bX^\top \bX)^{-1} \bX^\top \bGamma \big].
\end{equation*}
In this way, by construction, $\bGamma^*$ lies on the orthogonal complement of the column space of $\bX$, thus $\bX^\top \bGamma^* = \bzero$. 
Of course, we also need to transform $\bB$ to ensure that the identity $\bX \bB + \bGamma \bZ^\top = \bX \bB^{*\top} + \bGamma^* \bZ^\top$ holds true.

The same approach can be used to project $\bU$ and $\bV$ onto the orthogonal complement of the column space of $\bX$ and $\bZ$, respectively.
Then, an effective algorithm to obtain orthogonality with respect to the covariate column space is to sequentially project $\bGamma$, $\bU$, and $\bV$ onto the appropriate orthogonal complement, as shown in the first three rows of Algorithm~\ref{alg:projection_algorithm}.

Finally, to enforce one of (B1) or (B2), we can use standard reparametrizations based on the singular value decomposition of $\bU \bV^\top$ \citep[see, e.g.,][]{Risso2018, Liu2024}.
While (B3) requires first to rotate $\bU$ and $\bV$ using a whitening matrix, and then to triangularize $\bV$ using the QR decomposition \citep[see, e.g.,][]{Kidzinski2022}.
A detailed pseudo-code description of all the projection steps and the relative computational costs is outlined in Algorithm~\ref{alg:projection_algorithm}.

\begin{algorithm}[h!]
	\caption{%
		\label{alg:projection_algorithm}
		Pseudo-code description of the post-processing algorithm used to project an unrestricted estimate of the GMF model parameters onto the constrained space induced by the identifiability conditions (A), (B1), (B2), and (B3).
		On the right, we report the computational complexity of each step. 
		The compact notation \textrm{svd()} and \textrm{qr()} stand for the singular value and QR decomposition, respectively.
	}
	
	$\bD_{\scriptscriptstyle\Gamma} \gets (\bX^\top \bX)^{-1} \bX^\top \bGamma$; \quad
	$\bGamma \gets \bGamma - \bX \bD_{\scriptscriptstyle\Gamma}$; \quad
	$\bB \gets \bB + \bZ \bD_{\scriptscriptstyle\Gamma}^\top$; \hfill
	$O(p^3 + np^2 + npq)$ \\
	
	$\bD_{\scu} \gets (\bX^\top \bX)^{-1} \bX^\top \bU$; \quad
	$\bU \gets \bU - \bX \bD_{\scu}$; \quad
	$\bB \gets \bB + \bV \bD_{\scu}^\top$; \hfill
	$O(p^3 + np^2 + npd)$ \\
	
	$\bD_{\scv} \gets (\bZ^\top \bZ)^{-1} \bZ^\top \bV$; \quad
	$\bV \gets \bV - \bZ \bD_{\scv}$; \quad
	$\bGamma \gets \bGamma + \bU \bD_{\scv}^\top$; \hfill 
	$O(q^3 + mq^2 + mqd)$ \\
	
	\If{(B1)}{
		$\tilde\bU, \tilde\bSigma, \tilde\bV \gets \mathrm{svd}(\bU \bV^\top)$; \quad
		$\bU \gets \tilde\bU \tilde\bSigma$; \quad 
		$\bV \gets \tilde\bV$; \hfill
		$O(nmd +nd + md)$
	}
	\If{(B2)}{
		$\tilde\bU, \tilde\bSigma, \tilde\bV \gets \mathrm{svd}(\bU \bV^\top)$; \quad
		$\bU \gets \tilde\bU$; \quad 
		$\bV \gets \tilde\bV \tilde\bSigma$; \hfill
		$O(nmd + nd + md)$
	}
	\If{(B3)}{
		$\bS \gets \frac{1}{n} (\bU^\top\bU - \frac{1}{n} \bU^\top \bone_n \bone_n^\top \bU)$; \hfill 
		$O(nd^2 + nd + d^2)$ \\
		
		$\bU \gets \bU \bS^{-1/2}$; \quad
		$\bV \gets \bV \bS^{1/2}$; \hfill 
		$O(nd^2 + md^2 + d^3)$ \\
		
		$\bQ, \bR \gets \mathrm{qr}(\bV^\top)$; \quad
		$\bU \gets \bU \bQ$; \quad
		$\bV \gets \bR^\top$; \hfill 
		$O(md^2 + nd^2)$ \\
		
		$\bD \gets \diag(\bV)$; \quad
		$\bU \gets \bU \bD$; \quad
		$\bV \gets \bV \bD$; \hfill 
		$O(nd + md)$ \\
	}
\end{algorithm}

%%% APPENDIX A
\section{Additional algorithmic details}%
\label{app:algorithmic_details}

\subsection*{Stochastic gradient with non-zero covariate effects}
\label{app:stochastic_gradient_with_nonzero_covariate_effects}

In the general case where $\bB \neq \bzero$ and $\bGamma \neq \bzero$, in the optimization we must include an explicit update for $\bB$ and $\bGamma$. 
To this end, we define the first and second elementwise derivatives of the penalized deviance function with respect to $[\bGamma, \bU]$ and $[\bB, \bV]$ as
\begin{align*}
    \frac{\partial \ell_\lambda}{\partial [\bGamma, \bU]} = \dt{\bD} \,[\bZ, \bV] + \lambda [\bO, \bU], \qquad &
    \frac{\partial \ell_\lambda}{\partial [\bB, \bV]} = \dt{\bD}^\top [\bX, \bU] + \lambda [\bO, \bV], \\
    \frac{\partial^2 \ell_\lambda}{\partial [\bGamma, \bU]^2} = \ddt{\bD} \,[\bB * \bB, \bV * \bV] + [\bO, \bLambda], \qquad &
    \frac{\partial^2 \ell_\lambda}{\partial [\bB, \bV]^2} = \ddt{\bD}^\top [\bX * \bX, \bU * \bU] + [\bO, \bLambda].
\end{align*}
At the $t$th iteration of the proposed adaptive stochastic gradient descent algorithm, the above derivatives can be unbiasedly estimated by generalizing equations \eqref{eq:blockwise_minibatch_gradients}, thus obtaining 
\begin{align*}
    \hat\bG_{[\scgam, \scu],\row :}^t = (m / \ncol) \,\dt\bD_{\mb}^t [\bZ_{\col, :}, \bV_{\col :}^t] + \lambda [\bO, \bU_{\row :}^t], \qquad &
    \hat\bH_{[\scgam, \scu],\row :}^t = (m / \ncol) \,\ddt\bD_{\mb}^t [\bZ_{\col :} * \bZ_{\col :}, \bV_{\col :}^t * \bV_{\col :}^t] + [\bO, \bLambda], \\
    \hat\bG_{[\scb, \scv],\col :}^t = (n / \nrow) \,\dt\bD_{\mb}^{t \top} [\bX_{\row, :}, \bU_{\row :}^t] + \lambda [\bO, \bV_{\col :}^t], \qquad &
    \hat\bH_{[\scb, \scv],\col :}^t = (n / \nrow) \,\ddt\bD_{\mb}^{t \top} [\bX_{\row :} * \bX_{\row :}, \bU_{\row :}^t * \bU_{\row :}^t] + [\bO, \bLambda].
\end{align*}
As a result, the joint stochastic update for the regression parameters, $\bB$ and $\bGamma$, and latent variables, $\bV$ and $\bU$, is obtained as in \eqref{eq:blockwise_stochastic_quasi_newton_update}: 
\begin{align*}
    [\bGamma, \bU]_{\row :}^{t+1} \gets [\bGamma, \bU]_{\row :}^t + \rho_t \,\bDelta_{[\scgam, \scu], \row :}^t, \qquad & 
    \bDelta_{[\scgam, \scu], \row :}^t = - \alpha_t (\bar\bG_{[\scgam, \scu], \row :}^t \,\big/ \,\bar\bH_{[\scgam, \scu], \row :}^t), \\
    [\bB, \bV]_{\col :}^{t+1} \gets [\bB, \bV]_{\col :}^t + \rho_t \,\bDelta_{[\scb, \scv], \col :}^t, \qquad & 
    \bDelta_{[\scb, \scv], \col :}^t = - \alpha_t (\bar\bG_{[\scb, \scv], \col :}^t \,\big/ \,\bar\bH_{[\scb, \scv], \col :}^t).
\end{align*}
where $\bar\bG_{[\scgam, \scu],\row :}^t$, $\bar\bG_{[\scb, \scv],\col :}^t$, $\bar\bH_{[\scgam, \scu],\row :}^t$ and $\bar\bH_{[\scb, \scv],\col :}^t$ are the smoothed differentials computed via the exponential averaging in \eqref{eq:gradient_exponential_averaging}.

\subsection*{Unknown dispersion parameter}%
\label{app:unknown_dispersion_parameter}

In the case where the dispersion parameter $\phi$ is unknown and has to be learned from the data, a standard choice in the literature is the Pearson estimator, which is given by
\begin{equation*}
    \hat\phi 
    = \frac{1}{N} \sum_{i = 1}^{n} \sum_{j = 1}^{m} \frac{(y_{ij} - \hat\mu_{ij})^2}{\nu(\hat\mu_{ij}) / w_{ij}}
    = \frac{1}{N} \bone_n^\top \big[ \bW * (\bY - \hat\bmu)^2 / \nu(\hat\bmu) \big] \bone_m.
\end{equation*}
where $N = nm - mp - nq - (n+m)d - 1$ is the effective degrees of freedom of the model, that is the difference between the number of observations and the number of parameters to be estimated.
This can be computed \emph{a posteriori} or iteratively refined during the optimization substituting $\hat\bmu$ with $\bmu^t$.

In our optimization routine, we consider a sequential refinement of the dispersion parameter using a smoothed stochastic estimator obtained as
\begin{align*}
    \hat\phi^{t+1} & \gets \frac{1}{N}\frac{nm}{\nrow \ncol} \bone_{\nrow}^\top \big[ \bW_\mb * (\bY_\mb - \bmu_\mb^t)^2 / \nu(\bmu_\mb^t) \big] \bone_{\ncol}, \\
    \bar\phi^{t+1} & \gets (1 - \rho_t) \,\bar\phi^{t} + \rho_t \,\hat\phi^{t+1},
\end{align*}
where $\hat\phi^{t+1}$ is the stochastic estimate of $\phi$ obtained using only the information of the current mini-batch, while $\bar\phi^{t+1}$ is a smoothed estimator obtained as the exponential averaging of the current and previous estimates.

\subsection*{Negative Binomial inflation parameter}
\label{app:negative_binomial_inflation_parameter}

In the Negative Binomial model, the deviance and variance functions are specified as
\begin{equation*}
    D_\alpha (y, \mu) = 2 w \bigg[ y \log\frac{y}{\mu} - (y + \alpha) \log\frac{y + \alpha}{\mu + \alpha} \bigg], \quad
    \nu_\alpha (\mu) = w \mu (1 + \mu / \alpha),
\end{equation*}
where $\alpha$ is the shape parameter of the Negative Binomial family.
Since the shape parameter $\alpha > 0$ is typically unknown, we need to estimate it from the data. 
A common choice in the literature is to consider the moment estimator 
\begin{equation*}
    \hat\alpha 
    = \frac{\dsty \bigg[ \sum_{i=1}^{n} \sum_{j=1}^{m} w_{ij} \hat\mu_{ij}^2 \bigg]}{\dsty \bigg[ \sum_{i=1}^{n} \sum_{j=1}^{m} w_{ij} \big\{ (y_{ij} - \hat\mu_{ij})^2 - \hat\mu_{ij} \big\} \bigg] }
    = \frac{\bone_n^\top (\bW * \hat\bmu * \hat\bmu) \bone_m}{\bone_n^\top \big[ \bW * \{ (\bY - \hat\bmu)^2 - \hat\bmu \} \big] \bone_m},
\end{equation*}
where $\hat\bmu_{ij}$ must be a consistent estimator of the Negative Binomial mean.

Since complete access to the whole data and prediction matrices could be prohibitively expensive in high-dimensional settings, in our optimization scheme, we instead consider the stochastic update
\begin{align*}
    \hat\alpha^{t+1} & \gets \frac{\bone_{\nrow}^\top (\bW_\mb * \bmu_\mb^t * \bmu_\mb^t) \bone_{\ncol}}{\bone_{\nrow}^\top \big[ \bW_\mb * \{ (\bY_\mb - \bmu_\mb^t)^2 - \bmu_\mb^t \} \big] \bone_{\ncol}}, \\
    \bar\alpha^{t+1} & \gets (1 - \rho_t) \,\bar\alpha^{t} + \rho_t \max(\veps, \hat\alpha^{t+1}),
\end{align*}
where $\hat\alpha^{t+1}$ is the stochastic estimate of $\alpha$ obtained using only the information of the current mini-batch, while $\bar\alpha^{t+1}$ is a smoothed estimator obtained as the exponential averaging of the current and previous estimates, and $\veps > 0$ is a small positive constant introduced to ensure that the final estimate is positive.

\clearpage

%%% APPENDIX A
\section{Simulation setting details}%
\label{app:simulation_setting_details}

\subsection*{Data generating process}%
\label{subapp:data_generating_process}

To simulate the data, we use the \R~package \splatter~\citep{Zappia2017}, which is freely available on \bioconductor~\citep{Huber2015}.
In particular, we use the function \texttt{splatSimulateGroups()} to generate the gene-expression matrices.

In our experiments, we considered the following simulation setup: each dataset contains cells from five well-separated types evenly distributed in the sample.
The data are also divided into three batches having different expression levels.
No lineage or branching effects are considered.
The setting-specific dimensions of the gene-expression matrices, $n$ and $m$, are reported in the paper.
Under each simulation setting, we set the simulation parameters specifying the following options in the \splatter~functions \texttt{newSplatParams()} and \texttt{setParams()}:
\begin{itemize}
    \setlength\itemsep{-0.025cm}
    \item number of genes: \texttt{nGenes = $m$};
    \item number of cells: \texttt{nCells = $n$};
    \item number of cells per batch: \texttt{batchCells = c($\lfloor n/3 \rfloor$, $\lfloor n/3 \rfloor$, $n - 2 \lfloor n/3 \rfloor$)};
    \item probability of each cell-group: \texttt{group.prob = c(0.1, 0.2, 0.2, 0.2, 0.3)};
    \item probability of 
 gene differential expression in a group: \texttt{de.prob = c(0.3, 0.1, 0.2, 0.01, 0.1)};
    \item probability of gene down-regulation in a group: \texttt{de.downProb = c(0.1, 0.4, 0.9, 0.6, 0.5)};
    \item location of the differential expression factor: \texttt{de.facLoc = c(0.6, 0.1, 0.1, 0.01, 0.2)};
    \item Scale of the differential expression factor: \texttt{de.facScale = c(0.1, 0.4, 0.2, 0.5, 0.4)}.
\end{itemize}

\subsection*{Competing methods}%
\label{subapp:competing_methods}

We compare the proposed adaptive stochastic gradient descent method for the estimation of generalized matrix factorization models with several state-of-the-art approaches in the literature.
In particular, we consider the following models and algorithms.

\begin{description}
    \item \textbf{\cmf}: we use the \texttt{CMF()} function in the \cmfrec~package \citep{Cortes2023}, and we specify the following options: \texttt{k = $d$}, \texttt{nonneg = TRUE}, \texttt{user\_bias = FALSE}, \texttt{item\_bias = FALSE}, \texttt{center = FALSE}, \texttt{nthreads = 4}, \texttt{niter = 1000}.
    \item \textbf{\nmf}: we use the \texttt{NMF()} function in the \nmfpak~package \citep{Gaujoux2010}, and we specify the following options: \texttt{rank = $d$}, \texttt{method = "brunet"}, \texttt{seed = "nndsvd"}, \texttt{nrun = 1}.
    \item \textbf{\nmfp}: we use the \texttt{NNMF()} function in the \nnlm~package \citep{Lin2020}, and we specify the following options: \texttt{k = $d$}, \texttt{alpha = 1}, \texttt{beta = 1}, \texttt{n.threads = 4}, \texttt{method = "lee"}, \texttt{loss = "mkl"}, \texttt{max.iter = 1000}.
    \item \textbf{\avagrad}: we use the \texttt{glmpca()} function in the \glmpca~package \citep{Townes2019}, and we specify the following options: \texttt{L = $d$}, \texttt{fam = "poi"}, \texttt{minibatch = "none"}, \texttt{optimizer = "avagrad"}, \texttt{ctl = list(maxIter = 1000, tol = 1e-05)}.
    \item \textbf{\fisher}: we use the \texttt{glmpca()} function in the \glmpca~package \citep{Townes2019}, and we specify the following options: \texttt{L = $d$}, \texttt{fam = "poi"}, \texttt{minibatch = "none"}, \texttt{optimizer = "fisher"}, \texttt{ctl = list(maxIter = 200, tol = 1e-05)}.
    \item \textbf{\nbwave}: we use the \texttt{newFit()} function in the \NewWave~package \citep{Agostinis2022}, and we specify the following options: \texttt{K = $d$}, \texttt{commondispersion = TRUE}, \texttt{maxiter\_optimize = 200}, \texttt{stop\_epsilon = 1e-05}, \texttt{children = 4}.
    \item \textbf{\gfmam}: we use the \texttt{gfm()} function in the \GFM~package \citep{Liu2023}, and we specify the following options: \texttt{types = "poisson"}, \texttt{q = $d$}, \texttt{offset = FALSE}, \texttt{dc\_eps = 1e-05}, \texttt{maxIter = 200}, \texttt{algorithm = "AM"}.
    \item \textbf{\gfmvem}: we use the \texttt{gfm()} function in the \GFM~package \citep{Liu2023}, and we specify the following options: \texttt{types = "poisson"}, \texttt{q = $d$}, \texttt{offset = FALSE}, \texttt{dc\_eps = 1e-05}, \texttt{maxIter = 200}, \texttt{algorithm = "VEM"}.
    \item \textbf{\coap}: we use the \texttt{COAP\_RR()} function in the \COAP~package \citep{Liu2024}, and we specify the following options: \texttt{Z = $\bX$}, \texttt{q = $d$}, \texttt{epsELBO = 1e-05}, \texttt{maxIter = 100}, \texttt{joint\_opt\_beta = FALSE}, \texttt{fast\_svd = FALSE}.
    \item \textbf{\airwls}: we use the \texttt{cpp.fit.airwls()} function in the \sgdGMF~package, which implement the AIRWLS algorithm of \cite{Kidzinski2022} and \cite{Wang2023}, and we specify the following options: \texttt{ncomp = $d$}, \texttt{familyname = "poisson"}, \texttt{linkname = "log"}, \texttt{lambda = c(0,0,1,0)}, \texttt{maxiter = 200}, \texttt{nsteps = 1}, \texttt{stepsize = 0.2}, \texttt{eps = 1e-08}, \texttt{nafill = 1}, \texttt{tol = 1e-05}, \texttt{damping = 0.001}, \texttt{parallel = TRUE}, \texttt{nthreads = 4}.
    \item \textbf{\newton}: we use the \texttt{cpp.fit.newton()} function in the \sgdGMF~package, which implements the quasi-Newton algorithm of \cite{Kidzinski2022}, and we specify the following options: \texttt{ncomp = $d$}, \texttt{familyname = "poisson"}, \texttt{linkname = "log"}, \texttt{lambda = c(0,0,1,0)}, \texttt{maxiter = 200}, \texttt{stepsize = 0.2}, \texttt{eps = 1e-08}, \texttt{nafill = 1}, \texttt{tol = 1e-05}, \texttt{damping = 0.001}, \texttt{parallel = TRUE}, \texttt{nthreads = 4}.
    \item \textbf{\asgd}: we use the \texttt{cpp.fit.bsgd()} function in the \sgdGMF~package, and we specify the following options: \texttt{ncomp = $d$}, \texttt{familyname = "poisson"}, \texttt{linkname = "log"}, \texttt{lambda = c(0,0,1,0)}, \texttt{maxiter = 500}, \texttt{rate0 = 0.01}, \texttt{size1 = 100}, \texttt{size2 = 20}, \texttt{eps = 1e-08}, \texttt{nafill = 1}, \texttt{tol = 1e-05}, \texttt{damping = 1e-03}.
\end{description}

It is worth mentioning that the original \R~implementation of the \airwls~and \newton~methods can be found in the \GMF~package \citep{Kidzinski2022}. 
However, such an implementation does not permit the inclusion of gene-specific intercepts and covariate effects, say $\bgamma_i^\top \bsz_j$ in our notation, and also it does not allow for parallel computing in Windows operating systems. Therefore, we performed the benchmarking experiments using the \R/\Cpp~implementation in the proposed \sgdGMF~package.

All the options not specified here are left to default values.
For all the \R~scripts we used for running the simulation and plotting the results, please refer to the GitHub repository \href{https://github.com/alexandresegers/sgdGMF_Paper}{\texttt{github/alexandresegers/sgdGMF\_Paper}}.

%%% APPENDIX B
\section{Supplementary Figures and Tables}%
\label{app:supplementary_figures}
\setcounter{figure}{0}  

\renewcommand{\figurename}{Supplementary Figure}
\renewcommand{\tablename}{Supplementary Table}

\setcounter{table}{0}  

\begin{table}[H]
    \centering
    \caption{%
        \label{tab:exponential_family_distributions}%
        Exponential family laws along with their support space, canonical link ($g_c$), variance ($\nu$), and rescaled deviance ($D/2a$) functions.
        Here, we denote by $\alpha > 0$ the variance parameter of the Negative Binomial distribution.
    }
    \begin{tabular}{llccc}
        \toprule
        Distribution & Support & $g_c(\mu)$ & $\nu(\mu)$ & $D(y,\mu) / 2 a(\phi)$ \\
        \midrule
         Gaussian & $\real$ & $\mu$ & $1$ & $(y - \mu)^2 / 2$ \\
         Gamma & $\real_+$ & $1/\mu$ & $\mu^2$ & $(y - \mu) / \mu - \log(y / \mu)$ \\
         Inv. Gauss. & $\real_+$ & $1/\mu^2$ & $\mu^3$ & $(y - \mu)^2 / (2 y \mu^2)$ \\
         Poisson & $\natural_+$ & $\log(\mu)$ & $\mu$ & $y \log(y / \mu) - (y - \mu)$ \\
         Bernoulli & $\{0, 1\}$ & $\log\{ \mu / (1-\mu)\}$ & $\mu (1-\mu)$ & $y \log(y / \mu) + (1-y) \log \{ (1-y) / (1-\mu) \}$ \\
         Neg. Binom. & $\natural_+$ & $\log(\mu)$ & $\mu (1 + \mu / \alpha)$ & $y \log (y / \mu) - (y + \alpha) \log \{ (y+\alpha) / (\mu + \alpha) \}$ \\
         \bottomrule
    \end{tabular}
\end{table}

\begin{figure}[H]
    \centering
    \includegraphics[width = 
    \textwidth]{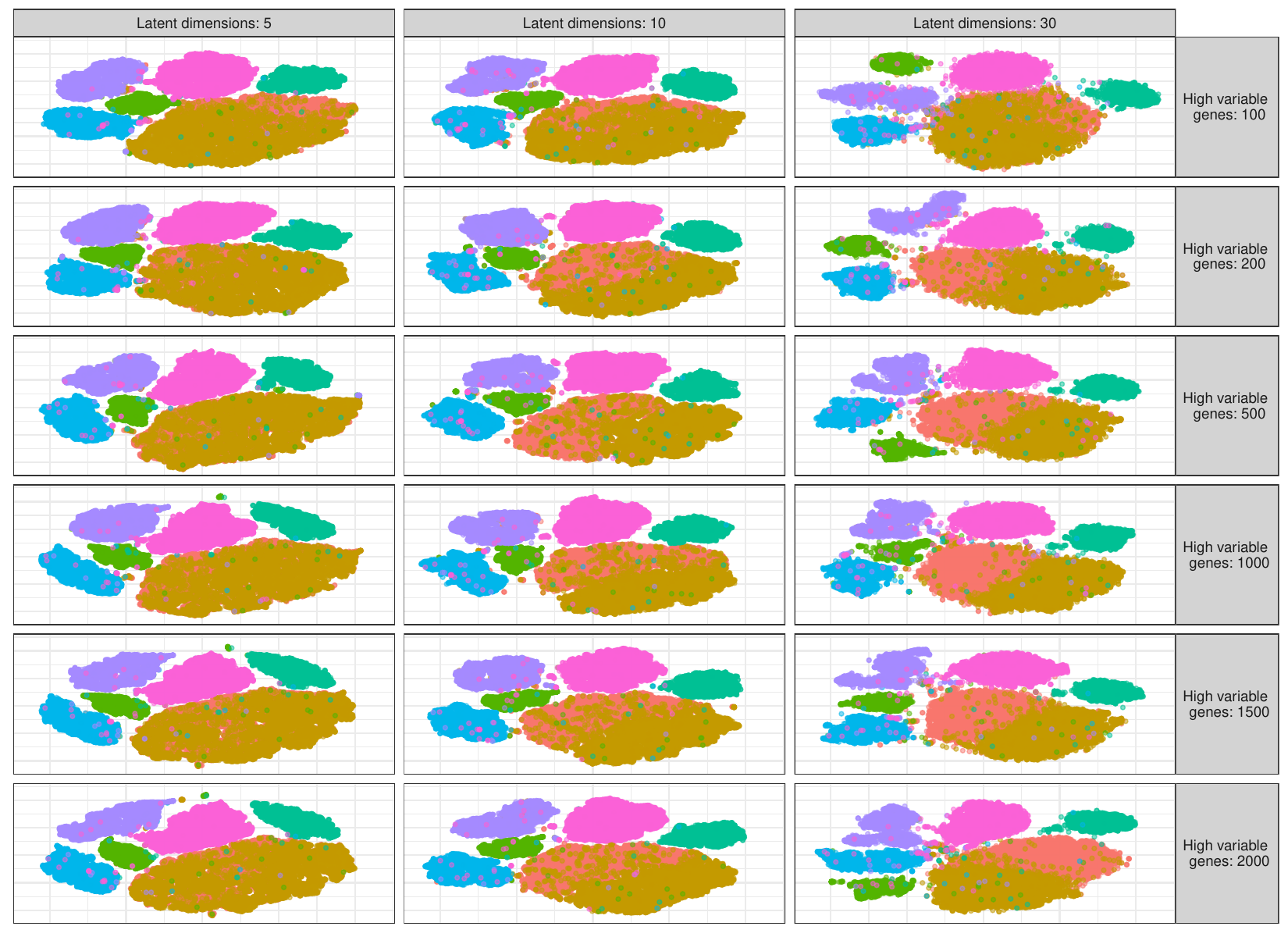}
    \caption{%
    	\label{fig:hvg}%
    	Comparison of tSNE embeddings using \asgd~with a varying number of highest variable genes (hvg) and a varying number of latent factors. No clear differences are observed between the different embeddings from using 500 hvg on, and therefore a default choice of 500 hvg is made.
    }
\end{figure}

\begin{figure}[H]
    \centering
    \includegraphics[width = 
    \textwidth]{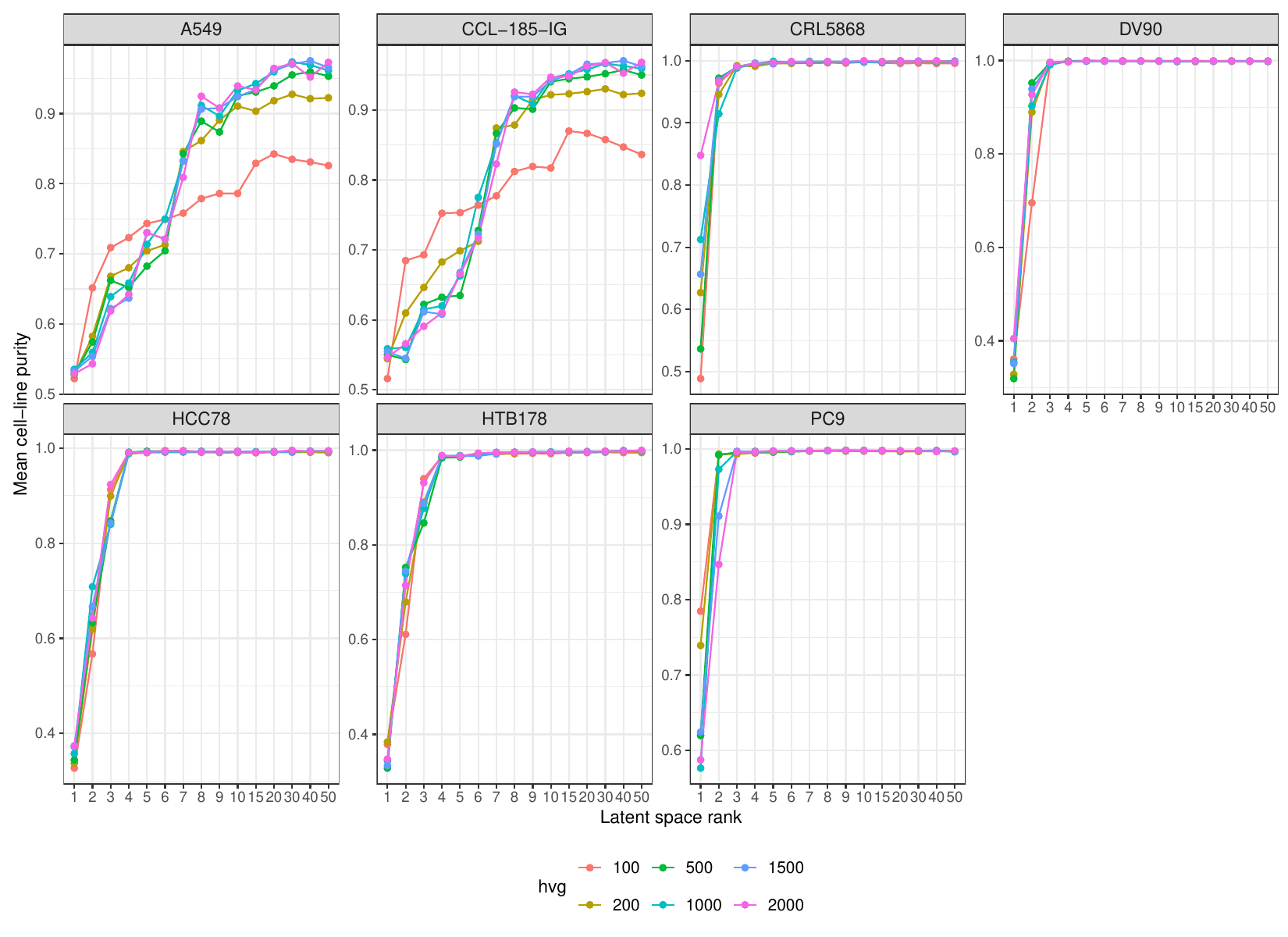}
    \caption{%
    	\label{fig:hvg_purity}%
    	Comparison of mean cell-line purity using \asgd~with a varying number of highest variable genes. No clear differences are observed between the different purities from using 500 hvg on, and therefore a default choice of 500 hvg is made.
    }
\end{figure}

\begin{figure}[H]
    \centering
    \includegraphics[width = 
    \textwidth]{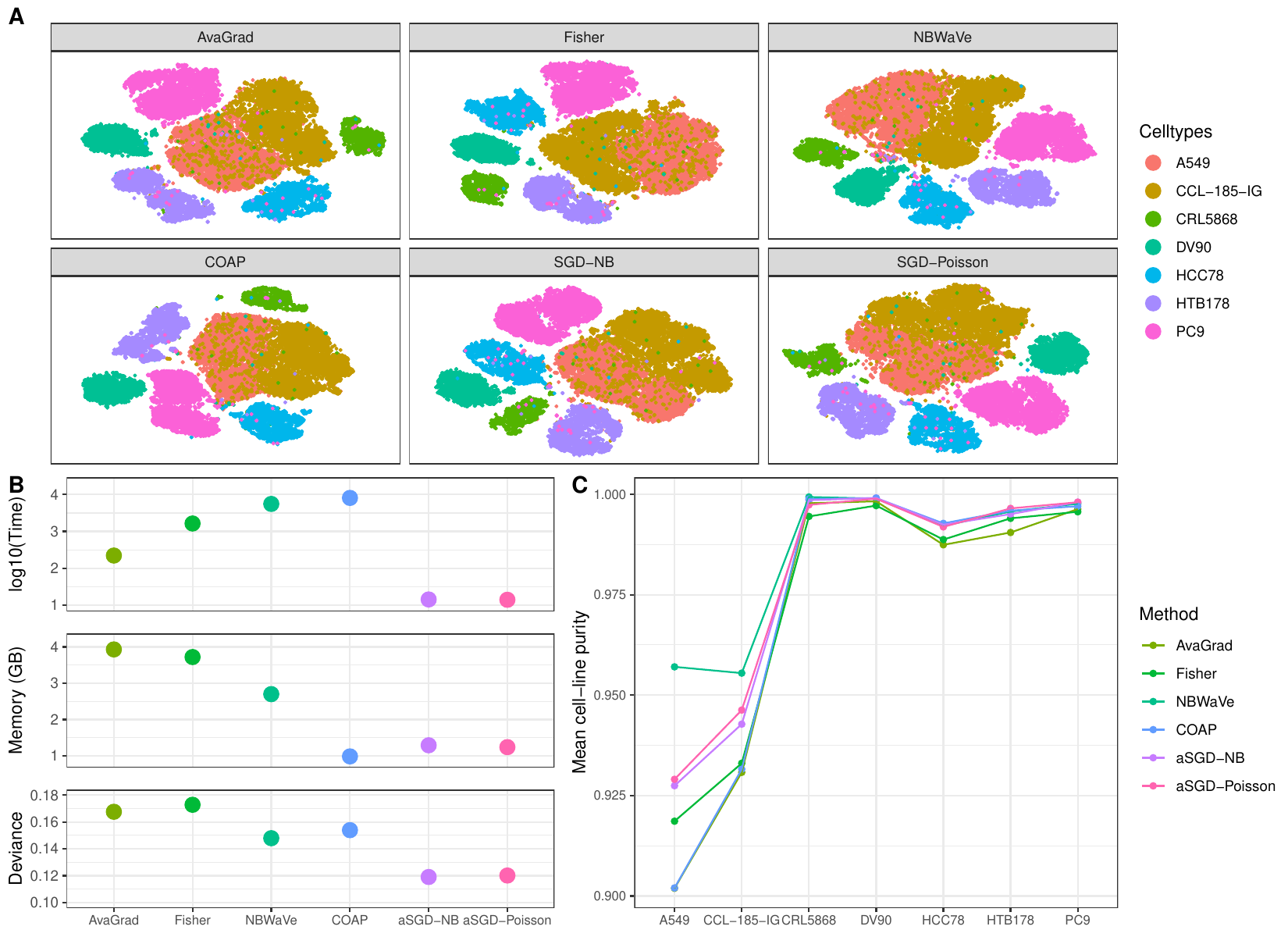}
    \caption{%
    	\label{fig:Arigoni-newwave-comparison}%
    	Comparison of \asgd~using block-subsampling with \nbwave~\citep{Agostinis2022}, \glmpca~\citep{Townes2019}, and \coap~\citep{Liu2024} on the Arigoni dataset \citep{Arigoni2024}. 
    	All methods use a matrix rank of 15, which was suggested by model selection criteria. 
    	A) tSNE embeddings show no clear differences between all methods. 
    	B) \asgd~is orders of magnitude faster compared to the other methods, and has a lower out-of-sample deviance prediction error when predicting missing values. 
    	This is probably due to \nbwave, \glmpca, and \coap~requiring imputation of missing values prior to computation of the latent structure, while \asgd~can deal with missing values internally. 
    	Also, RAM memory usage is similar between \asgd~and \coap, which are both much more efficient compared to \nbwave~and \glmpca.  
    	C) The mean cell-line cluster purity is also similar for all methods. 
    	\asgd~gives similar results on all levels when using the Poisson and Negative Binomial family on the Arigoni dataset. 
    	Therefore, \asgd~is a fast alternative to \nbwave, \glmpca, and \coap, while having similar performance.
    }
\end{figure}

\begin{figure}
    \centering
    \includegraphics[width = 
    \textwidth]{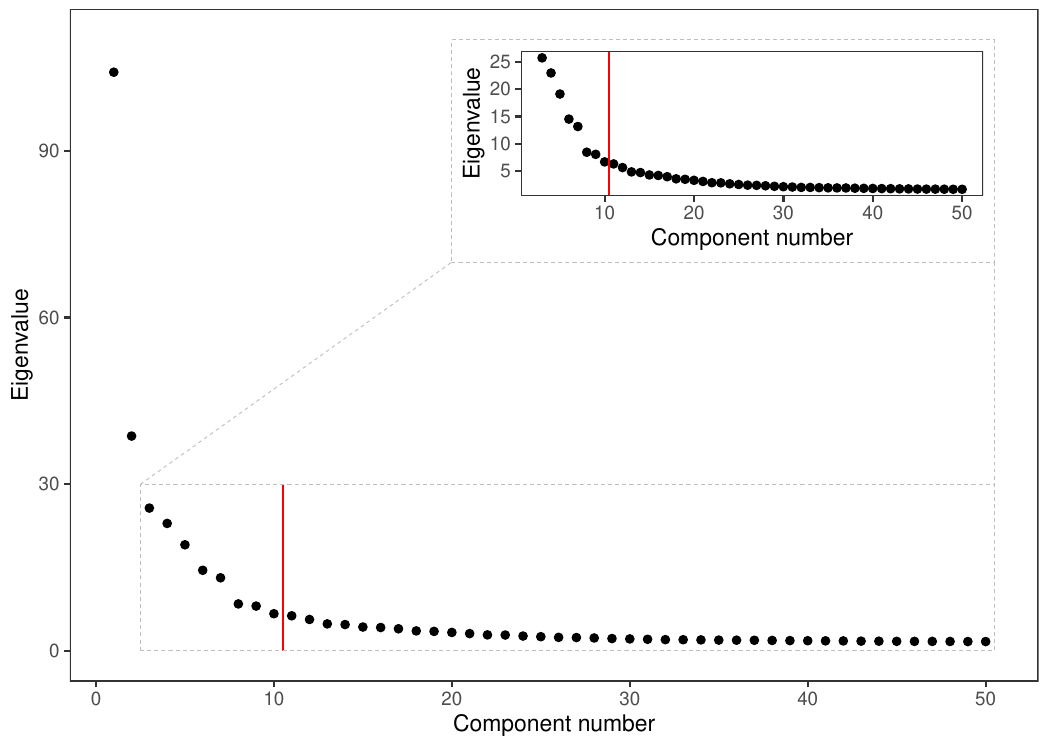}
    \caption{Eigenvalues of the latent directions ordered in descending values for the TENxBrainData \citep{Lun2023}. This plot, also called an elbow-plot can be used for model selection, by looking for an elbow in the curve. Here, 10 latent factors were chosen to be used in the case study.
    }
    \label{fig:casestudy-screeplot}
\end{figure}

\end{appendix}

%%% END OF DOCUMENT
\end{document}